\documentclass[11pt,a4paper]{article}

\usepackage{amsfonts, amsmath, amsthm, amssymb}
\usepackage{enumerate}
\usepackage{epsfig} 
\usepackage{graphicx}
\usepackage[T1]{fontenc} 
\usepackage[english]{babel}
\usepackage[applemac]{inputenc}

\usepackage{dsfont}
\usepackage{hyperref}

\usepackage[arrow, matrix, curve]{xy}

\newcommand{\M}{\mathcal P}
\newcommand{\IR}{\mathbb R}
\newcommand{\IP}{\mathbb P}
\newcommand{\IE}{\mathbb E}
\newcommand{\IN}{\mathbb N}

\newcommand{\dd}{\mathrm{d}}
\newcommand{\dbl}[1]{d_{BL}(#1)}
\newcommand{\weak}{\rightharpoonup}
\newcommand{\supp}{\mathrm{supp}\,}

\newcommand{\Div}{\mathrm{div}}

\theoremstyle{plain}
\newtheorem{Theorem}{Theorem}[section]
\newtheorem{Lemma}[Theorem]{Lemma}
\newtheorem{Corollary}[Theorem]{Corollary}
\newtheorem{Proposition}[Theorem]{Proposition}
\theoremstyle{definition}
\newtheorem{Definition}[Theorem]{Definition}
\newtheorem{Remark}[Theorem]{Remark}
\newtheorem{Remarks}[Theorem]{Remarks}

\begin{document}
\title{A particle approximation for the relativistic Vlasov-Maxwell dynamics}
\author{Dustin Lazarovici\thanks{Mathematisches Institut, Ludwig-Maximilians-Universit\"at, Theresienstr. 39, 80333 Munich, Germany. E-Mail: lazarovici@math.lmu.de}}
\date{}
\maketitle
\begin{abstract} \noindent We present a microscopic derivation of the 3-dimensional relativistic Vlasov-Maxwell system as a combined mean field and point-particle limit of an $N$-particle system of rigid charges with $N$-dependent radius. The approximation holds for typical initial particle configurations, implying in particular propagation of chaos for the respective dynamics.\\

\end{abstract}
	
\section{Introduction}
\noindent We are interested in a microscopic derivation of the three dimensional relativistic Vlasov-Maxwell System. This is a set of partial differential equations describing a collisionless plasma of identical charged particles interacting through a self-consistent electromagnetic field:

\begin{equation}\begin{split}\label{MV}
		&\partial_t f + v(\xi)\cdot \nabla_x f + K(t,x,\xi)\cdot \nabla_\xi f = 0, \\
		&\partial_t E - \nabla_x \times B = -j, \hspace{.6cm}  \nabla_x\cdot E = \rho,\\
		&\partial_t B + \nabla_x \times E = \;0,  \hspace{.8cm}\nabla_x \cdot B =0.
	\end{split}\end{equation}
	
	\noindent Here, units are chosen such that all physical constants, in particular the speed of light, are equal to $1$. The distribution function $f(t,x,\xi) \geq 0$ describes the density of particles with position $x \in \IR^3$ and relativistic momentum $\xi \in \IR^3$. The other quantities figuring in the Vlasov-Maxwell equations are the relativistic velocity of a particle with momentum $\xi$, given by 
	\begin{equation}
		v(\xi) = \frac{\xi}{\sqrt{1+\lvert\xi\rvert^2}},
	\end{equation}
	\noindent and the charge and current density entering Maxwell's equations, given by
	\begin{equation}
		\rho(t,x) = \int\limits f(t,x,\xi) \,\mathrm{d}{\xi}, \hspace{4mm} j(t,x) = \int v(\xi) f(t,x,\xi)\, \mathrm{d}{\xi}. 
	\end{equation}
	
	\noindent The function
	\begin{equation}
		K(t,x,\xi) = E(t,x) + v(\xi) \times B(t,x)
	\end{equation}
	\noindent thus describes the Lorentz force acting at time $t$ on a particle at $x$ moving with momentum $\xi$.\\

	While the Vlasov-Maxwell equations have been successfully applied in pfor a long time, their microscopic derivation is still an open problem. In the electrostatic (nonrelativistic) case, important results were obtained by Hauray and Jabin \cite{HaurayJabin}, who were able to prove \emph{mean field limits} for singular forces -- up to but not including the Coulomb case -- with an $N$-dependent cut-off in the case of strong singularities (and without cut-off for force kernels diverging slower than $\frac{1}{\lvert x \rvert}$ at the origin). Coulomb interactions were recently included in \cite{Dustin} and \cite{PeterDustin}, with cut-offs decreasing as $N^{-1/3 + \epsilon}$ and $N^{-1/15 +\epsilon}$, respectively, amounting to a particle approximation for the \emph{Vlasov-Poisson} equation. 
	
	The aim of this paper is to combine and generalize the methods into a \cite{Dustin} and \cite{PeterDustin} into a microscopic derivation of the 3-dimensional relativistic Vlasov-Maxwell system. The mean field limit for Vlasov-Maxwell is considerably more complex, as it involves relativistic (retarded) interactions and the electromagnetic field as additional degrees of freedom. However, we will show that the basic insights and techniques developed for the Vlasov-Poisson equation can be extended to the relativistic regime. 
	
	As a microscopic theory, we consider an $N$-particle system of extended, rigid charges, also known as the \emph{Abraham model} (after \cite{Abraham}, see \cite{Spohn} for a discussion). Size and shape of the particles are described by an $N$-dependent form factor that approximates a $\delta$-distribution in the limit $N \to \infty$. The cut-off parameter thus has a straightforward physical interpretation in terms of a finite electron-radius. Our approximation of the Vlasov-Maxwell dynamics will thus be a combination of mean field limit and point-particle limit, similar to the result in \cite{Dustin} where we treated the non-relativistic limit.

A previous result for the Vlasov-Maxwell system was obtained by Golse \cite{Golse}, who uses an equivalent regularization with fixed (but arbitrarily small) cut-off to derive a \emph{mollified} version of the equations (i.e. the smearing persists in the limiting equation). This is analogous to the pioneering work of Braun and Hepp, Dobrushin and Neunzert, wo treated non-relativistic interactions with Lipschitz continuous force kernel. As Golse notes (see \cite[Prop. 6.2]{Golse}), his result can be applied to approximate the actual Vlasov-Maxwell system but only in a very weak sense, basically corresponding to choosing an $N$-dependent cut-off decreasing as $\sim \log(N)^{-\frac{1}{2}}$. In the spirit of the recent developments in the Vlasov-Poisson case, will considerably improve upon this result, allowing the cut-off to decrease as $N^{-\frac{1}{12}}$.

	\subsection{Structure of the paper}
	The paper is structured as follows: 
	\begin{enumerate}[]
		\item We will first recall a representation of the electromagnetic field in terms of Li\'{e}nard-Wiechert distributions that was derived, for instance, in \cite{Bouchut}. The key advantage of this representation is that it does not depend on derivatives of the current-density, thus allowing for better control of fluctuations in terms of the Vlasov density. 
		
		\item In Section \ref{Section:microscopictheory}, we introduce the Abraham model of rigid charges as our microscopic theory and define a corresponding regularized mean field equation. By introducing an appropriate $N$-dependent rescaling, we will take the mean field limit together with a point-particle limit, in which the electron-radius goes to $0$ and the particle form factor approximates a $\delta$-distribution. This will allow us to approximate the actual Vlasov-Maxwell dynamics in the large $N$ limit. 
		
		\item In Section \ref{section:existenceofsolutions} we recall some known results about existence of (strong) solutions to the Vlasov-Maxwell equations. 
		
		\item After stating our precise results in Section \ref{section:statementofresults}, we derive a few simple but important corollaries from the solutions theory of the Vlasov-Maxwell equations in Section \ref{section:Corfromsol}.
		
		\item In Section \ref{Section:strategyofproof}, we will follow the method developed in \cite{Peter} and \cite{PeterDustin} and introduce a stochastic process that will serve as our ``measure of chaos'', quantifying the difference between mean field dynamics and microscopic dynamics.

		\item In Section \ref{Section:globalestimates} we derive some global bounds on the (smeared) microscopic charge density and the corresponding fields. 
		
		\item Section \ref{section:pointwiseestimates} then contains the more detailed law-of-large number estimates for the difference between mean field dynamics and microscopic dynamics. These estimates are derived from the Li\'{e}nard-Wiechert decomposition of the fields and are somewhat similar to the bounds proven in \cite{Bouchut} for the regularity of solutions. 
		
		\item Finally, we combine all estimates into a proof of the mean field limes for the Vlasov-Maxwell dynamics (Section \ref{section:Gronwall}). We end with some remarks regarding the obtained results and the status of the microscopic regularization (Section \ref{section:remarks}).
	\end{enumerate}
	
	
	
	\section{Field representation}
	The Vlasov-Maxwell system contains in particular Maxwell's equations
	\begin{equation}
		\begin{split}
			\partial_t E - \nabla_x \times B = -j, \hspace{.6cm}  \nabla_x\cdot E = \rho,\\
			\partial_t B + \nabla_x \times E = \;0,  \hspace{.8cm}\nabla_x \cdot B =0,
		\end{split}\end{equation}
		where charge- and current-density are induced by the Vlasov density $f(t,x,\xi)$. In general, Maxwell's equations can be solved by introducing a scalar potential $\Phi$ and a vector potential $A$, satisfying
		\begin{equation}
			\square_{t,x} \Phi = \rho, \hspace*{1cm} \square_{t,x} A = j,
		\end{equation}
		in terms of which the electric and magnetic fields are given by
		\begin{align} 
			E(t,x) = - \nabla_x \Phi(t,x) - \partial_t A(t,x); \hspace{.6cm} B(t,x) = \nabla \times A(t,x).
		\end{align}
		It is convenient to split the potential into a homogeneous and an inhomogeneous part, i.e. 
		$A = A_0 + A_1$ with
		\begin{align} \label{Ahomo}
			\square_{t,x} A_0 &= 0, \; \partial_t A_0 \mid_{t=0} = -E_{in}\\\label{Ainhomo}
			\square_{t,x} A_1 &= j, \; A_1\mid_{t=0} = \partial_t A_1\mid_{t=0} = 0.
		\end{align}
		
		\noindent We recall that the retarded fundamental solution of the d'Alembert  operator $\square_{t,x} = \partial^2_t - \Delta_x$ (in $3+1$ dimensions) is given by the distribution
		
		\begin{equation}\label{dAlembert}
			Y(t,x) = \frac{\mathds{1}_{t > 0}}{4 \pi t} \delta(\lvert x \rvert - t).
		\end{equation}
		
		\noindent Hence, in the Vlasov-Maxwell system, a solution of \eqref{Ainhomo} is given by
		\begin{equation}
			A_1 = Y*_{t,x} j = \int v(\xi) Y*_{t,x} f(\cdot, \cdot, \xi ) \dd \xi.
		\end{equation}
		Similarly, we set
		\begin{equation}
			\Phi= \Phi_1 =  Y*_{t,x} \rho =  \int  Y*_{t,x} f(\cdot, \cdot, \xi ) \dd \xi.
		\end{equation}
		
		\noindent The solution of the homogeneous wave-equation is given by (see e.g. \cite[Thm. 4.1]{Jalal})
		\begin{equation}\label{homopotential}
			A_0(t,\cdot) = Y(t, \cdot)*_x E_{in},
		\end{equation}
		where the initial field has to satisfy the constraint
		\begin{equation}\label{Eincompatibility}
			\Div E_{in} = \rho_0 = \int f(0, \cdot, \xi) \dd \xi.
		\end{equation}
		Hence, 
		\begin{equation}\label{DefEin} E_{in} =- \nabla_x G *_x \rho_0 + E'_{in}\end{equation}
		with
		\begin{equation} G(x)= \frac{1}{4\pi \lvert x \rvert}, \; x \in \IR^3, \; \text{and } \Div  E'_{in}= 0.\end{equation} 
		
		\noindent In total, for a given distribution function $f_t$, the Lorentz force-field $K(t,x,\xi) = E(t,x) + v(\xi) \times B(t,x)$ is given by
		\begin{align}\label{lorentzforce}
			K[f] = & -  \int \partial_t \nabla_x \,  (Y(t,\cdot) *_x G *_x f_0(\cdot, \eta)) \dd \eta\\\label{lorentzforce2}
			&- \int (\nabla_x + v(\eta) \partial_t )\, Y*f(\cdot,\cdot,\eta) \dd \eta\\\label{lorentzforce3}
			& -  \int v(\xi) \times (v(\eta) \times \nabla_x) \, Y*f(\cdot,\cdot,\eta) \dd \eta,
		\end{align}
		where we have set $E'_{in} = 0$, for simplicity. In more detail, this formulation of the field equations can be found e.g. in \cite{Golse}. Note that equations (\ref{lorentzforce} - \ref{lorentzforce3}) still allow for various representation in terms of $f$, depending on how one evaluates the derivatives. 
		
		\subsection{Li\'enard-Wiechert distributions}\label{section:lienard}
		A particularly useful representation of the electromagnetic field can be given as a superposition of Li\'enard-Wiechert  fields (see, in particular, \cite[Lemma 3.1]{Bouchut}.)
		For a given distribution  $f_t$, the induced electric field can be written as
		\begin{equation*} E(t,x)= E_0(t,x) + E'_0(t,x)+E_1(t,x)+E_2(t,x) \end{equation*}
		where 
		\begin{align}\label{E0term}
			E_0[f_0] & = \;\;\; - \partial_t Y(t, \cdot) *_x E_{in}\\\label{E0'term}
			E'_0[f_0]&= \;\;\;\int (\alpha^0 Y)(t, \cdot, \xi) *_{t,x} f_0  \, \dd \xi \\\label{E1term}
			E_1[f] &= \;\;\; \int (\alpha^{-1} Y) *_{t,x} (\mathds{1}_{t \geq 0} f) \, \dd \xi \\\label{E2term}
			E_2[f] &= - \int (\nabla_\xi \alpha^0 Y) *_{t,x} (K \mathds{1}_{t \geq 0} f )\,  \dd \xi 
		\end{align}
		\noindent with
		\begin{align}\label{integralkernels}
			\alpha^0(t,x,\xi)  = \frac{x - t v(\xi)}{t - v(\xi) x};  \;\;\;  \alpha^{-1}(t,x,\xi) = \frac{(1-v(\xi)^2) (x - t v(\xi))}{(t-v(\xi)x)^2}.
		\end{align}
		Hence 
		\begin{equation}\label{nablaalpha}
			(\nabla_\xi \alpha^0)^i_j (t, x, \xi) = \frac{t (t-v\cdot x)(v_jv^i - \delta^i_j) + (x_j - t v_j)(x^i - (v \cdot x) v^i)}{\sqrt{1+ \lvert \xi \rvert^2} (t - v \cdot x)^2}.
		\end{equation}
		Here, we follow the notation from \cite{Bouchut}; The upper index in $\alpha^{j}, \, j = 0, -1,$ refers to the degree of homogeneity in $(t,x)$. 
		\begin{itemize}
			\item[] $E_2$ is called the radiation or acceleration term. It dominates in the far-field and depends on the acceleration of the particles.
			
			\item[] $E_1$ corresponds to a relativistic Coulomb term and grows like the inverse square distance in the vicinity of a point source. 
			
			\item[] $E'_0$ are ``shock waves'', depending only on the initial data and propagating with speed of light (c.f. \cite{Vera}).
			
			\item[] $E_0$ is the homogeneous field generated by the potential \eqref{homopotential}. It depends only on $E_{in}$ and thus on the initial charge distribution via the constraint \eqref{Eincompatibility}.
		\end{itemize}
		
		\noindent Similar expressions hold for the magnetic field. One finds that
		\begin{equation*} B(t,x)= B_0(t,x) + B'_0(t,x)+B_1(t,x)+B_2(t,x) \end{equation*}
		\noindent with
		\begin{align}
			B'_0[f_0] &= \int (n \times \alpha^0 Y)(t, \cdot, \xi) *_x f_0 \, \dd \xi \\
			B_1[f] &= \int (n \times \alpha^{-1} Y) *_{t,x} (\mathds{1}_{t \geq 0} f) \, \dd \xi \\
			B_2[f] &= - \int (\nabla_\xi (n \times \alpha^0 Y)) *_{t,x} (K \mathds{1}_{t \geq 0} f )\,  \dd \xi
		\end{align}
		where we introduced the normal vector $ n(x) := \frac{x}{\lvert x \rvert}$.
		\begin{Remark}
			In the physical literature, the Li\'enard-Wiechert field is usually written in terms of the particle acceleration $\dot{v}$ rather than the force $\dot{\xi}$. Since $v(\xi) = \frac{\xi}{\sqrt{1+\lvert \xi \rvert^2}}$, the two expressions are related as $\dot{v} = \sqrt{1 - \lvert v \rvert^2} (K - (v\cdot K) v)$.
		\end{Remark}

		\section{Microscopic theory (Abraham model)}\label{Section:microscopictheory}
		Consider a system of $N$ identical point-charges with phase-space trajectories $(x_i(t), \xi_i(t))_{i=1,..,N}$. The corresponding charge- and current-densities are then given by 
		\begin{equation}
			\label{microdensity}
			\rho(t,x) = \sum\limits_{i=1}^N \delta (x - x_i(t)) ; \;\;\; j(t,x)=\sum\limits_{i=1}^N v(\xi_i(t))\delta (x - x_i(t))\end{equation}
		and generate an electromagnetic field $(E,B)(t,x)$ according to Maxwell's equations. However, together with the Lorentz-force equation
		\begin{align}\label{Lorentzforceeq}
			\begin{cases}\frac{\dd}{\dd t}{x}_{i} (t) = v(\xi_i(t))\\[1.5ex]
				\frac{\dd}{\dd t} {\xi}_{i}(t) = E(t, x_i(t)) + v(\xi_i(t)) \times B(t,x_i(t)) \end{cases}
		\end{align}
		
		\noindent this does not yield a consistent theory due to the \emph{self-interaction singularity}: The fields generated by \eqref{microdensity} are singular precisely at the location of the particles, where they would have to be evaluated according to \eqref{Lorentzforceeq}.

		A classical way to regularize the Maxwell-Lorentz theory is to consider instead of point-particles a system of extended, rigid bodies to which the charge is permanently attached. This is also known as the Abraham model. Shape and size of the rigid charges are given by a smooth, compactly supported, spherically symmetric \emph{form factor} $\chi$ satisfying: 
		\begin{equation}\label{formfactor} \chi  \in C^\infty_c(\IR^3); \; \chi(x) = \chi(\lvert x \rvert); \; \chi(x) = 0 \text{ for } \lvert x \rvert > r=1;  \; \int \chi (x) \, \dd x = 1. \end{equation}
		
		\noindent The corresponding charge- and current-densities are then given by
		\begin{equation}\label{RCdensities}\rho(t,x) = \frac{1}{N} \sum\limits_{i=1}^N \chi(x- x_i(t)); \;\;\;  j(t,x)=\frac{1}{N} \sum\limits_{i=1}^N v(\xi_i(t))\chi (x - x_i(t)), \end{equation}
		
		\noindent where $x_i(t)$ now denotes the center of mass of particle $i$. In order to approximate the Vlasov-Maxwell  equations, we shall perform the mean field limit together with a point-particle limit, introducing an $N$-dependent electron-radius $r_N$ which tends to zero as $N \to \infty$. We thus define a \emph{rescaled form factor} $\chi^N$ by
		\begin{equation}\label{rescalledformfactor}\chi^{N}(x):= r_N^{-3}\, \chi\Bigl(\frac{x}{r_N}\Bigr), \; N \in \IN, \end{equation}
		where $(r_N)_N$ is a decreasing sequence with $r_N =1, \lim\limits_{N \to \infty} r_N = 0$, to be specified later. This rescaled form factor satisfies
		\begin{equation} \lVert \chi^N \rVert_\infty = r_N^{-3}; \; \chi^N(x) = 0 \text{ for } \lvert x \rvert > r_N;  \; \int \chi^N (x) \, \dd x = 1 \end{equation}
		and approximates a $\delta$-measure in the sense of distributions. 
		
		\noindent In the so-called \emph{mean field scaling}, the new field equations read
		\begin{equation}\label{Abraham1}
			\begin{cases}
				\partial_t E - \nabla_x \times B = - \frac{1}{N} \sum\limits_{i=1}^N v(\xi_i(t))\chi^N (x - x_i(t)),\\
				\nabla_x\cdot E =  \frac{1}{N} \sum\limits_{i=1}^N \chi^N(x- x_i(t)),\\[1.2ex]
				\partial_t B + \nabla_x \times E = \;0,  \hspace{.8cm}\nabla_x \cdot B =0. \end{cases}
		\end{equation}
		
		\noindent The particles move according to the equation of motion
		\begin{equation}\label{Abraham2}
			\begin{cases}\frac{\dd}{\dd t}{x}_{i} (t) = v(\xi_i(t))\\[1.5ex]
				\frac{\dd}{\dd t} {\xi}_{i}(t) = \int \chi^N (x-x_i(t)) \bigl [ E(t,x) + v(\xi_i(t)) \times B(t,x) \bigr ] \, \dd x. \end{cases}\\
		\end{equation}
		
		\noindent An equivalent regularization was used by Rein \cite{Rein2} to prove the existence of weak solutions to the Vlasov-Maxwell equations, and by Golse \cite{Golse} to prove the mean field limit for the regularized Vlasov-Maxwell system. For any fixed $r_N$, initial particle configuration $Z=(x_i, \xi_i)_{1\leq i \leq N}$ and initial field configuration $(E_{in}, B_{in}) \in C^2(\IR^3)$ satisfying the constraints
		\begin{equation} \Div E_{in}(x) = \frac{1}{N}\sum\limits \chi^N(x-x_i), \; \Div B_{in}(x) = 0, \end{equation}
		the system of equations defined by \eqref{Abraham1} and \eqref{Abraham2} has a unique strong solution as proven in \cite{GernotDirk} and \cite{KomechSpohn}.
		
		Note that the Abraham model is only semi-relativistic, because the charges are assumed to maintain their shape in any frame of reference, neglecting the relativistic effect of Lorentz-contraction. Rotations of the rigid particles are neglected, as well (though one may expect that these degrees of freedom can be separated anyway due to spherical symmetry of the form factor). On the other hand, one important virtue of this theory is that the total energy
		\begin{equation}\label{totalenergy} \varepsilon = \frac{1}{N} \sum\limits_{i=1}^N \sqrt{1 + \lvert \xi_i(t) \rvert^2 } + \frac{1}{2} \int E^2(t,x) + B^2(t,x)\, \dd x \end{equation} 
		is a constant of motion, as we will verify with a simple computation. 
		
			

		\subsection{The regularized Vlasov-Maxwell  system}
		
		In view of the extended charges model defined by equations \eqref{Abraham1} and \ref{Abraham2}, we introduce a corresponding mean field equation. For a given form factor $\chi \in C^\infty_c$ and a rescaling sequence $(r_N)_N$ , we consider the set of equations
		\begin{equation}\begin{split}\label{RMV}
				&\partial_t f + v(\xi)\cdot \nabla_x f + \widetilde K(t,x,\xi)\cdot \nabla_\xi f = 0, \\
				&\partial_t E - \nabla_x \times B = - \tilde j, \hspace{.6cm}  \nabla_x\cdot E = \tilde \rho,\\
				&\partial_t B + \nabla_x \times E = \;0,  \hspace{.8cm}\nabla_x \cdot B =0.
			\end{split}\end{equation}
			\begin{equation}\label{regularizeddensity}
				\tilde \rho =\chi^N *_x \int\limits f (t, \cdot ,\xi) \,\mathrm{d}{\xi}, \hspace{4mm} \tilde j =\chi^N *_x \int v(\xi) f(t, \cdot ,\xi)\, \mathrm{d}{\xi}.
			\end{equation}
			
			\begin{equation}
				\widetilde K(t,x,\xi) = \chi^N *_x \bigl(E + v(\xi) \times B \bigr)(t,x)
			\end{equation}
			
			\noindent where $\chi^N$ is the rescaled form factor defined in \eqref{rescalledformfactor}. We call this set of equations the \textit{regularized Vlasov-Maxwell system} with cut-off parameter $r_N$.

			Since the $L^1$ norm of $\rho$ propagates along any local solution and $\lVert D^\alpha \tilde \rho_t \rVert_\infty \leq \lVert D^\alpha \chi^N \rVert_\infty \lVert \rho_t \rVert_1$ all spatial derivatives of $\tilde \rho$ and $\tilde j$ are bounded uniformly in time. This is enough to show global existence of classical solutions for compact initial data $f_0 \in C^1_c(\IR^3 \times \IR^3),  \tilde E_{in}, \tilde B_{in} \in C^2_c(\IR^3)$ satisfying the constraints $\Div \tilde E_{in} = \tilde \rho_0, \, \Div  \tilde B_{in} = 0$, see \cite{Rein, Horst2} for more details. 
			
			According to the method of characteristics (see e.g. \cite{Golse}) $(\Psi_{t,0}(Z), E, B)$ is a solution of the Abraham model \eqref{Abraham1}, \ref{Abraham2} with initial data $(Z, E_{in}, B_{in})$ if and only of $(\mu^N[\Psi_{t,0}(Z)], E, B)$ is a solution of the regularized Vlasov-Maxwell system \eqref{RMV} in the sense of distributions with initial data $(\mu^N_0[Z], E_{in}, B_{in})$. 
			
			\begin{Remark} The regularized Vlasov-Maxwell system defined above is not exactly the same as the one considered by Golse \cite{Golse} or Rein \cite{Rein}, at least not a priori. In those publications, a double convolution is applied to the charge/current density, that is, the fields solve Maxwell's equation for $\rho = \chi^N*\chi^N*\int f(t,\cdot, \xi) \dd \xi, \; j=  \chi^N*\chi^N*\int v(\xi)f(t,\cdot, \xi) \dd \xi$. Here, only one mollifier is used in \eqref{regularizeddensity} to regularize the charge/current density, a second convolution with $\chi^N$ is applied as the fields act back on $f_t$, mirroring the form of the rigid charges model defined by eqs. (\ref{Abraham1},\ref{Abraham2}). However, by using the uniqueness of solutions to Maxwell's equation and the fact that convolutions commute with each other and with derivatives, one checks that both formulations of the regularized Vlasov-Maxwell dynamics are actually equivalent.
			\end{Remark}


			\section{Existence of solutions}\label{section:existenceofsolutions}
			While the 3-dimensional Vlasov-Poisson equation is very well understood from a PDE point of view, the state of research is less satisfying when it comes to the Vlasov-Maxwell equations. Existence of global weak solutions was first proven in DiPerna, Lions, 1989 \cite{DiPernaLions}. Concerning existence and uniqueness of classical solutions, no conclusive answer has been given, so far. The central result is the paper of Glassey and Strauss, 1986, aptly titled ``singularity formation in a collisionless plasma could occur only at high velocities''  \cite{GlasseyStrauss}. We recall their main theorem in the following.
			
			\begin{Theorem}[Glassey-Strauss, 1986]\label{Thm:Glassey}
				Let $f_0 \in C^1_c(\IR^3 \times \IR^3)$ and $E_{in}, B_{in} \in C^2_c(\IR^3)$ satisfying $ \Div E^{in} = \rho[f_0], \; \Div B_0 = 0$. Let $(f_t, E_t, B_t)$ be a (weak) solution of the Vlasov-Maxwell System \eqref{MV} with initial datum $(f_0,E_{in}, B_{in})$. Suppose there exists $T\in [0, +\infty]$ and $C > 0$ such that 
				\begin {equation}\label{momentumbound} R(t) = \sup \lbrace \lvert \xi \rvert : \exists x \in \IR^3 \, f(t,x,\xi) \neq 0 \rbrace < C, \; \forall t < T \end{equation}
			Then:
			\begin{equation}  \sup_{0\leq t < T^*} \lbrace \lVert f_t\rVert_{W^{1,\infty}_{x,\xi}},  \lVert (E_t, B_t) \rVert_{W^{1,\infty}_{x}} \rbrace <  \infty \end{equation}
			where $ \lVert f \rVert_{W^{1,\infty}_{x,\xi}} = \lVert f \rVert_\infty + \lVert \nabla_{x,\xi} f \rVert_\infty$ etc.   Hence, $(f_t, E_t, B_t)$  is the unique classical solution on $[0,T)$ with initial data $(f_0,E_{in}, B_{in})$.
		\end{Theorem}
		\noindent Simply put, the theorem states that singularity formation can occur in finite time only if particles get accelerated to velocities arbitrarily close to the speed of light. Subsequently, seemingly weaker conditions have been identified that ensure the boundedness of the momentum support and thus the existence of strong solutions. For instance, Sospedra-Alfonso and Illner \cite{Illner} prove: 
		
		\begin{equation}\label{Illner}\limsup\limits_{t \to T^-}\, R(t) = + \infty \;\; \Rightarrow \;\; \limsup\limits_{t \to T^-}\, \lVert \rho[f_t] \rVert_\infty = +\infty. \end{equation}
		
		\noindent Most recently, Pallard \cite{Pallard} showed that 
		
		\begin{equation}\label{Pallard}\limsup\limits_{t \to T^-}\, R(t) = + \infty \;\; \Rightarrow \;\; \limsup\limits_{t \to T^-}\, \lVert \rho[f_t] \rVert_{L^6(\IR^3)} = +\infty. \end{equation}
		
		\noindent Unfortunately, the criteria thus established are still far away from the known a priori bounds (the strongest, in $L^p$-sense, being the kinetic-energy bound on $\lVert \rho [f_t] \rVert_{L^{4/3}(\IR^3)}$, see e.g. \cite{Rein}) so that well-posedness of the Vlasov-Maxwell system is still considered an open problem. Note that the conditions \eqref{Illner} and \eqref{Pallard} are actually necessary and sufficient for \eqref{momentumbound}, because $ \rho_t(x)  = \int f(t,x,\xi) \dd \xi \leq \frac{4\pi}{3} R^3(t) \lVert f_0 \rVert_\infty$. \\
		
		\noindent We will also need the following theorem of Rein \cite{Rein2}, who used the regularization introduced above to establish the existence of global weak solutions to the Vlasov-Maxwell system, simplifying the original proof of DiPerna and Lions \cite{DiPernaLions}. 
		
		\begin{Theorem}[Rein, 2004]\label{Thm:Rein}
			Let $f_0 \in L^1\cap L^\infty(\IR^3\times\IR^3)$ and $E_{in}, B_{in} \in L^2(\IR^3)$ satisfying the compatibility condition \eqref{compatibility}. Let $(f^N_t, E^N_t, B^N_t)$ be a solution of the regularized Vlasov-Maxwell  system \eqref{RMV} with initial data $(f_0, \tilde E_{in}, \tilde B_{in})$. Then there exist functions $f \in L^\infty(\IR; L^1 \cap L^\infty(\IR^6)),  E,B\in L^\infty(\IR;L^2(\IR^3))$ such that, along a subsequence,
			\begin{align*} f^{N} \weak f \; in\; L^\infty([0,T]\times \IR^6);\;\; E^{N},B^{N} \weak E, B \; in \; L^2([0,T]\times \IR^3), k \to \infty \end{align*}
			\noindent for any bounded time-interval $[0, T], \, T >0$ and $(f, E, B)$ is a global weak solution of the Maxell-Vlasov system \eqref{MV} with $\lim\limits_{t\to 0} (f_t, E_t, B_t) = (f_0, E_{in}, B_{in})$ and $\lVert f_t \rVert_{L^p(\IR^6)} = \lVert f_0 \rVert_{L^p(\IR^6)}$ for all $p \in [1, \infty], \,  t>0$. 
		\end{Theorem}

		\section{Statement of the results}\label{section:statementofresults}
		
		\noindent In the previous sections, we have introduced three kinds of dynamics: The Vlasov-Maxwell system \eqref{MV}, the regularized Vlasov-Maxwell system \eqref{RMV} and the microscopic Abraham model of extended charges (\ref{Abraham1},\ref{Abraham2}). In order to approximate one solution by the other, it does not suffice to assume that the respective distributions are (in some sense) close at $t=0$. We also have to fix the incoming fields in an appropriate manner, otherwise free fields can be responsible for large deviations between mean field dynamics and microscopic dynamics. We will note our respective convention in the following definition.

		\begin{Definition}\label{Convention}
			Let $f_0 \in C^1_c(\IR^3 \times \IR^3)$ with $f_0 \geq 0, \int f_0(x,\xi) \dd x \dd \xi = 1$ and $E_{in}, B_{in} \in C^2_c(\IR^3)$ satisfying the Gauss constraints
			\begin{equation}\label{compatibility} \Div E_{in} = \rho[f_0] = \int f_0 (\cdot, \xi) \dd \xi, \;\;  \Div B_{in} = 0.\end{equation}
			Such $(f_0, E_{in}, B_{in})$ are the admissible initial data for the Vlasov-Maxwell system \eqref{MV}.
			\begin{enumerate}[1)]
				\item  For the regularized Vlasov-Maxwell system, we fix initial data for the fields as
				\begin{equation}\label{macroin} E^{N}_{in} := \chi^N*E_{in}, \;\;\;  B^N_{in} := \chi^N*B_{in}, \end{equation}
				for any $N \geq 1$. These fields satisfy: $ \Div  E^N_{in} = \tilde \rho[f_0]$ and $ \Div B^N_{in} = 0$. We denote by $(f^N, E^N, B^N)$ the unique solution of \eqref{RMV} with initial data $(f_0, E^N_{in}, B^N_{in})$. 	
				\item For the microscopic system with initial configuration $Z =(x_1, \xi_1, ..., x_N,\xi_N) \in \IR^{6N}$, the charge distribution can be written as $ \tilde\rho[\mu^N[Z]](x) = \frac{1}{N}\sum\limits_{i=1}^N \chi^N(x-x_i)$. Given a renormalizing sequence $(r_N)_{N \geq 1}$ we fix compatible initial fields $(E^{\mu}_{in}, B^\mu_{in})$ such that
				\begin{equation}\label{microin} E^\mu_{in} :=  E^N_{in} - \nabla G *( \tilde\rho[\mu^N_0[Z]] - \tilde\rho[f_0]), \;\;\;  B^\mu_{in} := B^N_{in}. \end{equation}
				Note that $E^\mu_{in}$ and $B^\mu_{in}$ depend on $N$ and $E^\mu_{in}$ also on $Z$. For any $N \in \IN$ and $Z=(x_i,\xi_i) \in \IR^{6N}$ we then denote by $\bigl((x_i^*, \xi_i^*)_{1\leq i \leq N}, E^\mu, B^\mu\bigr)$ the unique solution of $(\ref{Abraham1}, \ref{Abraham2})$ with initial data $(Z, E^\mu_{in}, B^\mu_{in})$. We call
				\begin{equation} {}^{N}\Psi_{t,0} = \IR^{6N} \to  \IR^{6N}, \; {}^{N}\Psi_{t,0}(Z)=(x_i^*(t), \xi_i^*(t) )_{i=1,..,N} \end{equation}
				the \emph{microscopic flow} and 
				\begin{equation} \mu^N_t[Z] := \mu^N[\Psi_{t,0}(Z)] = \frac{1}{N}\sum\limits_{i=1}^N \delta_{x^*_i(t)}  \delta_{\xi^*_i(t)}
				\end{equation}
				the \emph{microscopic density} of the system with initial configuration $Z$.
			\end{enumerate}
		\end{Definition} 
		\noindent Note: The macroscopic fields $(E^N_{in}, B^N_{in})$ are compactly supported, though the microscopic field $E_{in}^\mu$, determined by \eqref{macroin}, is not.\\
		

		\noindent We now state our precise result in the following theorem. Our approximation of the Vlasov-Maxwell dynamics is formulated in terms of the Wasserstein distances $W_p$ that play a central role in the theory of optimal transportation and that were first introduced in the context of kinetic equations by Dobrushin.  We shall briefly recall the definition and some basic properties. For further details, we refer the reader to the  book of Villani \cite[Ch. 6]{Villani}.
	\begin{Definition}
		\noindent Let $\M(\IR^k)$ the set of probability measures on $\IR^k$ (equipped with its Borel algebra). For given $\mu, \nu \in \M(\IR^k)$ let $\Pi(\mu, \nu)$ be the set of all probability measures $\IR^k \times \IR^k$ with marginal $\mu$ and $\nu$ respectively. 
		
		\noindent For $p\in [1,\infty)$  we define the \emph{Wasserstein distance} of order $p$ by
		
		\begin{equation}
		W_p(\mu, \nu) := \inf\limits_{\pi \in \Pi(\mu,\nu)} \, \Bigl( \int\limits_{\IR^k\times\IR^k} \lvert x -  y \rvert^p \, \dd \pi(x,y) \, \Bigr)^{1/p}.   
		\end{equation}
		
		\noindent Convergence in Wasserstein distance implies, in particular, weak convergence in $\M(\IR^k)$, i.e. 
		\begin{equation*} \int \Phi(x)\, \dd \mu_n(x) \to \int \Phi(x)\, \dd \mu(x), \;\;\; n \to \infty, \end{equation*}
		for all bonded, continuous functions $\Phi$. Moreover, convergence in $W_p$ implies convergence of the first $p$ moments. $W_p$ satisfies all properties of a metric on $\M(\IR^k)$, except that it may take the value $+\infty$. 
	\end{Definition}	 
	
	\noindent An important result is the \emph{Kantorovich-Rubinstein duality}:
	\begin{equation}\begin{split}\label{Kantorovich}W^p_p(\mu, \nu) =  \sup \Bigl\lbrace &\int \Phi_1(x) \, \dd\mu(x) - \int \Phi_2(y) \, \dd\nu(y) : \\ 
	&(\Phi_1, \Phi_2) \in L^1(\mu)\times L^1(\nu), \Phi_1(y) - \Phi_2(x) \leq \lvert x - y \rvert^p \Bigr \rbrace. \end{split} \end{equation}
	
	\noindent A particularly useful case is the first Wasserstein distance, for which the problem reduces further to
	\begin{equation*} W_1(\mu, \nu) = \sup\limits_{\lVert \Phi \rVert_{Lip} \leq 1}\Bigl\lbrace \int \Phi(x) \, \dd \mu(x) -  \int \Phi(x)\, \dd \nu(x) \Bigr\rbrace, \end{equation*}
	where $\lVert \Phi \rVert_{Lip}:= \sup\limits_{x\neq y} \frac{\Phi(x)-\Phi(y)}{\lvert x - y \rvert}$, to be compared with the \emph{bounded Lipschitz distance}
	\begin{equation*} \dbl{\mu, \nu} = \sup \Bigl\lbrace \int \Phi(x) \, \dd \mu(x) -  \int \Phi(x)\, \dd \nu(x)\, ;\; \lVert \Phi \rVert_{Lip}, \lVert \Phi \rVert_\infty \leq 1\Bigr\rbrace.\end{equation*}

		 \noindent In the following, probabilities and expectation values referring to initial data $Z \in \IR^{6N}$ are meant with respect to the product measure $\otimes^N f_0$ for a given probability density $f_0 \in L^1\cap L^\infty(\IR^3 \times \IR^3)$. That is, for any random variable $H:\mathbb{R}^{6N}\to\mathbb{R}$ and any element $A$ of the Borel-algebra  we write
		\begin{align}
			\mathbb{P}^N_0(H\in A)=& \int_{H^{-1}(A)} \prod_{j=1}^N f_0(z_j)dZ,\\
			\mathbb{E}^N_t(H)=& \int_{\mathbb{R}^{6N}} H(Z)\prod_{j=1}^N f_0(z_j)dZ\;.
		\end{align}
		When the particle number $N$ is fixed, we will usually omit the index and write only $\IP_0$, respectively $\IE_0$.

		\begin{Theorem}\label{Thm:Thm3}
			Let $f_0 \in C^1_c(\IR^3 \times \IR^3, \IR^+_0)$ with total mass one and $(E_{in}, B_{in}) \in C^2_c(\IR^3)$ satisfying the constraints \eqref{compatibility}.  Let $\gamma < \frac{1}{12}$ and $r_N$ a rescaling sequence with $r_N \geq N^{-\gamma}$. For $N \in \IN$, let $(f^N, E^N, B^N)$ the solution of the renormalized Vlasov-Maxwell equation \eqref{RMV} and $(\Psi_{t,0}(Z), E^\mu, B^\mu)$ the solution of the microscopic equations (\ref{Abraham1} \ref{Abraham2}) with initial data as in Def. \ref{Convention}. Let  $\mu^N_t[Z] := \mu^N[\Psi_{t,0}(Z)]$ the empirical density corresponding the the microscpic flow $\Psi_{t,0}(Z)$. Suppose there exists $T > 0$ and constant $C_0 > 0$ such that
			\begin{equation}\label{Assumption2} \lVert \rho[f^N_t] \rVert_\infty \leq C_0, \; \forall N \in \IN, \, 0 \leq t \leq T. \end{equation}
			\begin{enumerate}[a)]
				\item Then we have molecular chaos in the sense that for all $p \in [1, \infty)$ and $\epsilon > 0$: 
				\begin{equation} \forall  0 \leq t \leq T: \; \lim\limits_{N \to \infty} \IP^N_0 \Bigl[ W_p(\mu^N_t[Z], f_t) \geq \epsilon \Bigr] = 0 \end{equation}
				where $(f_t,E_t, B_t)$ is the unique classical solution of the Vlasov-Maxwell system \eqref{MV} on $[0,T]$ with initial data $(f_0, E_{in}, B_{in})$.
				
				\item For the regularized dynamics, we have the following quantitative approximation result: Let $p \geq 1$, $\alpha < \min \lbrace \frac{1}{6}, \frac{1}{2p} \rbrace$ and $\gamma < \delta < \frac{1}{4}$. Then there exist constants $L , C$ depending on $T, C_0$ and the initial data such that for all $t \in [0,T]$ and $N \geq 4$:
				\begin{equation}\label{Thm1}\IP_0 \Bigl[ \sup\limits_{0 \leq s \leq t} W_p(\mu^N_s[Z], f^N_s) \geq N^{- \delta} + e^{tL}N^{-\alpha} \Bigr] \leq e^{tC \sqrt{\log(N)}} N^{-\frac{1}{4} + \delta} + a(N,p,\alpha) \end{equation}
				where
				\begin{equation} 
					a(N,p,\alpha) = c' \cdot  \begin{cases}\exp(-cN^{1-2p \alpha}) & \text{if } p > 3\\
						\exp(-c\frac{N^{1- 6 \alpha}}{\log(2+N^{3\alpha})^{2}}) &  \text{if } p = 3\\
						\exp(-cN^{1- 6 \alpha}) &\text{if } p \in [1,3). \end{cases}\end{equation}
				The constant $c', c > 0$ depend only on $p, \alpha$ and $f_0$.\\

				\item For the fields, we have the following approximation results: For any compact region $M \subset \IR^3$ there exists a constant $C_1 > 0$ such that for any $0 \leq t \leq T$ and $N \geq 4$:
				\begin{equation} \begin{split}\IP_0 \Bigl[ \lVert (E^N_t, B^N_t) - (E^\mu_t, B^\mu_t)\rVert_{L^\infty(M)} \geq C_1 \sqrt{\log(N)} N^{-\delta} \Bigr]
						\leq e^{tC\sqrt{\log(N)}} N^{-\frac{1}{4} + \delta} .
					\end{split}\end{equation} 
					
				\end{enumerate}
				
			\end{Theorem}
			
			\begin{Remarks}\mbox{} 
				\begin{enumerate}[1)]
					\item The result implies propagation of molecular chaos in the sense of convergence of marginals.
					
					
					\item We do not have a quantitative result for the convergence $f^N_t \weak f_t$, i.e. we do not know how fast $W_p(f^N_t, f_t)$ converges to $0$ for any $p$.

					\item Assumption \eqref{Assumption2} can be replaced by equivalent conditions, e.g. a uniform bound on $ \lVert \rho[f^N_t] \rVert_{L^6(\IR^3)}$ or on the momentum-support. Of course, it would be much more desirable to have a sufficient condition on $f_0$ only. However, such a condition would likely have to come out of the existence theory for Vlasov-Maxwell.  
					
					\item The constants $C$ and $C_0$ blow up as the maximal velocity $\overline v$ approaches 1 (speed of light).	
				\end{enumerate}
			\end{Remarks}
			
			\section{Corollaries from solution theory}\label{section:Corfromsol}
			We will first conclude some corollaries from the existence theorems cited above. Fix $f_0 \in C^1_c(\IR^3 \times \IR^3, \IR^+_0)$ and $T > 0$ as in Theorem \ref{Thm:Thm3}. By assumption, there exists $C_0$ such that 
			\begin{equation} \lVert \rho[f^N_t] \rVert_\infty \leq C_0, \; \forall N \geq 1, \, 0 \leq t \leq T. \end{equation}
			By the theorem of Sospedra-Alfonso and Illner \cite{Illner}, there thus exists a $\mathcal{R} > 0$ such that
			\begin{equation}\label{DefR} R[f^N](t) = \sup \lbrace \lvert \xi \rvert : \exists x \in \IR^3 \, f^N(t,x,\xi) \neq 0 \rbrace < \mathcal{R}, \end{equation}
			for all $N \geq 1$ and  $0 \leq t \leq T$. We define
			\begin{equation} \overline{\xi} := \mathcal{R} + 1 \text{ and  } \overline{v} := \lvert v(\overline{\xi})\rvert,\end{equation}
			which will serve us as an upper bound on the velocity of the particles. By the Glassey-Strauss theorem, there thus exists a constant $L' > 0$ such that
			\begin{equation}
				\lVert (E^N_t,B^N_t) \rVert_{\infty} + \lVert \nabla_x (E^N_t,B^N_t) \rVert_{\infty} \leq {L'}, 
			\end{equation}
			for all $N \geq 1, \, 0 \leq t \leq T$. In particular, observing that
			\begin{equation}
				\nabla_\xi v(\xi) = \nabla_\xi \frac{\xi}{\sqrt{1+\xi^2}} = \frac{\delta^{i,j}}{\sqrt{1+\xi^2}} - \frac{\xi^i \xi^j}{(\sqrt{1+\xi^2})^3}, \end{equation}
			with $\lvert \nabla_\xi v(\xi) \rvert \leq 2$, we have
			\begin{equation}\label{LipschitzL} \lVert K[f^N](t, \cdot, \cdot) \rVert_{W^{1,\infty}(\IR^3\times \IR^3)} \leq \max\lbrace L', 2 \rbrace =:L.
			\end{equation}  
			
			\noindent Note that the theorems of Glassey/Strauss und Sospedra-Alfonso/Illner  are formulated for the unregularized Vlasov-Maxwell system \eqref{MV}, so one has to check that they actually yield bounds that are uniform in $N$ as one considers the sequence of regularized solutions $f^N_t$. We refer, in particular, to the simplified proof of the Glasey-Strauss theorem proposed by Bouchut, Golse and Pallard \cite{Bouchut}. For instance, the $W^{1,\infty}$-bound on the fields is derived from estimates of the form 
			\begin{align*} 
				\lVert K (t) \rVert_{W^{1,\infty}_{x,\xi}} &\leq C_2e^{TC_2} \bigl( 1 + \log_+(\lVert\nabla_{x} f \rVert_{L^\infty([0,T] \times \IR^3 \times \IR^3)})\bigr),\\
				\sup\limits_{s \leq t}\lVert \nabla_{x,\xi} f(s) \rVert_\infty &\leq  \lVert \nabla_{x,\xi} f_0 \rVert_\infty + C_1 \int\limits_{0}^t (1 +  \log_+(\sup\limits_{s' \leq s}\lVert\nabla_{x,\xi} f(s') \rVert_\infty)) \sup\limits_{s' \leq s}\lVert \nabla_{x,\xi} f(s') \rVert_\infty \dd s,
			\end{align*}
			where $\log_+(x) := \max \lbrace 0 , \log(x) \rbrace$ and the constants $C_1, C_2 $ depend only on $T, f_0$ and $\mathcal{R}$ (see \cite[Section 5.4]{Bouchut}). Hence, one readily sees that the bounds hold independent of $N$. \\
			
			\noindent Since the velocity of the particles is bounded by $1$, the support in the space-variables remains bounded, as well, for compact initial data. We set 
			\begin{equation}\label{xsupport} \overline{r}  = \sup \bigl \lbrace \lvert x \rvert : \exists \xi  \in \IR^3 \, f_0(x,\xi) \neq 0 \bigr\rbrace + T + 1. \end{equation} 
			Then we have, in particular, $\supp{\tilde \rho[f_t]} \subseteq \mathrm{B}(\overline{r}; 0) = \lbrace x \in \IR^3 : \lvert x \rvert \leq \overline{r} \rbrace$ for all $0 \leq t \leq T$ as well as $\lvert \Psi^1_{t,0}(Z) \rvert_\infty < \overline{r}$ if $Z \in \supp \otimes^N f_0$.\\
			
			\noindent Now we recall from Theorem \ref{Thm:Rein} that, along a subsequence, 
			\begin{equation} (f^N_t, E^N, B^N) \weak (f'_t, E'_t, B'_t), \end{equation}
			where $(f', E', B')$ is a global weak solution of the Vlasov-Maxwell system \eqref{MV} with initial data $(f_0, E_{in}, B_{in})$ and weak convergence of the fields is understood in $L^2$ sense. However, for any $t \in [0,T]$ and any test-function $\varphi \in  C^\infty_c(\IR^3 \times \IR^3)$ with $\lvert \xi \rvert < \mathcal{R} \Rightarrow \varphi(x , \xi) = 0$,
			\begin{equation*} \int \varphi(x,\xi) f'_t(x,\xi) \dd \xi \dd x = \lim\limits_{N \to \infty} \int \varphi(x,\xi) f^N_t(x,\xi) \dd \xi \dd x = 0.\end{equation*} This means that the momentum-support of $f'$ remains bounded by $\mathcal{R}$ and according to the Glassey-Strauss theorem, $(f', E', B')$ is actually a strong solution on $[0,T]$. Thus, under the assumptions of the theorem, we can conclude that
			\begin{equation}  (f^N_t, E^N_t, B^N_t) \weak (f_t, E_t, B_t), \; \forall 0 \leq t \leq T,\end{equation}
			where $(f_t, E_t, B_t)$ is the  \emph{unique classical} solution on $[0,T]$ with initial data  $(f_0, E_{in}, B_{in})$ and the convergence holds for any subsequence (otherwise one could extract a convergent subsubsequence) and thus for the sequence itself.
			
			Finally, note that since we can restrict all measures to the compact space $\mathrm{B}(\overline{r})\times \mathrm{B}(\overline \xi) $, weak convergence is equivalent to convergence in Wasserstein distance so that, in particular, $W_p(f^N_t, f_t) \to 0$ for all $p \in [1, \infty)$.
			
			\section{Strategy of proof}\label{Section:strategyofproof}
			
			\begin{Definition}\label{Def:macroflow}
				Let $f_0, E_{in}, B_{in}$ as above. Let $f^N_t$ the solution of the regularized Vlasov-Maxwell  system with initial datum $f_0$. Let $K[\tilde f^N]$ the Lorentz-force field corresponding to the charge- and current-density induced by $\tilde f^N = \chi^N * f^N$. We denote by $\varphi^N_{t,s}$ the characteristic flow of the regularized Vlasov-Maxwell system \eqref{RMV}, i.e. the solution of  
				\begin{equation} 
					\begin{cases}\label{meanfieldeq}\frac{\dd}{\dd t}{y} (t) = v(\eta(t))\\[1.5ex]
						\frac{\dd}{\dd t} {\eta}(t) = \tilde K[\tilde f^N](t, y, \eta) \end{cases}\end{equation}
				with $\varphi^N_{s,s}(z) = z$. We denote by ${}^{N}\Phi_{t,s}$ the lift of $\varphi^N_{t,s}(\cdot)$ to the $N$-particle phase-space, that is ${}^N\Phi_{t,s}(Z) := (\varphi^N_{t,s}(z_1), ... , \varphi^N_{t,s}(z_N))$. In other words, ${}^N\Phi_{t,s}$ is the $N$-particle flow generated by the (regularized) mean field force induced by $f^N_t$. We will often omit the index $N$. 
				
			\end{Definition}

			\noindent Our result is based on the method of Boers and Pickl \cite{Peter} that was recently refined in \cite{PeterDustin} in the context of Vlasov-Poisson. We introduce the following quantity as a measure of molecular chaos. 
			
			\begin{Definition}\label{Def:J2}
				Let ${}^N\Phi_{t,0}$ the mean field flow defined above and ${}^N\Psi_{t,0}$ the microscopic flow solving \eqref{Abraham2}. We denote by ${}^N\Psi^1_{t,0} = (x^*_i(t))_{1\leq i \leq N}$ and ${}^N\Psi^2_{t,0} = (\xi^*_i(t))_{1\leq i \leq N}$ the projection onto the spatial, respectively the momentum coordinates. 
				
				\noindent Let $J(t)$ be the stochastic process given by
				\begin{equation}\begin{split}J^N_t(Z) := \min \Bigl\lbrace 1, \lambda(N) N^{\delta}\sup\limits_{0 \leq s \leq t} \lvert {}^N\Psi^1_{t,0}(Z) -  {}^N\Phi^1_{t,0}(Z) \rvert_\infty&\\ +  N^\delta \sup\limits_{0 \leq s \leq t} \lvert  {}^N\Psi^2_{t,0}(Z) -  {}^N\Phi^2_{t,0}(Z) \rvert_\infty &\Bigr\rbrace, \end{split} \end{equation}
				where $\lvert Z \rvert_\infty = \max \lbrace \lvert x_i \rvert :  1\leq i\leq N \rbrace$ denotes the maximum-norm on $\IR^{3N}$ and $\lambda(N) := \max \lbrace 1 ,  \sqrt{\log(N)} \rbrace$. 
			\end{Definition}
			
			\noindent Our aim is to derive a Gronwall estimate for the time-evolution of $\IE^N_0(J^N_t)$, showing that $\IE^N_0(J^N_t) \xrightarrow{N \to \infty} 0, \, \forall 0 \leq t \leq T$. This will be achieved by using the Li\'{e}nard-Wiechert representation of the fields introduced in section \ref{section:lienard}. The field corresponding to the (regularized) Vlasov-Maxwell dynamics is generated by the smeared Vlasov-density $\tilde f^N$, while the field corresponding to the microscopic dynamics of the rigid charges is generated by the smeared microscopic density $\tilde \mu^N[Z] := \chi^N*_x \mu[Z]$. For a given space-time point $(t,x) \in \IR \times \IR^3$, we will estimate the difference as:
			\begin{align}\notag &\bigl\lvert E_i[\tilde f^N](t,x) - E_i[\tilde \mu^N](t,x) \bigr\rvert \\\label{macroterm0}
				&\leq \bigl\lvert E_i[\tilde f^N](t,x) - E_i[\tilde \mu^N [\Phi_{s,0}(Z)]](t,x) \bigr\rvert\\\label{microterm0}
				&+  \bigl \lvert E_i[\tilde \mu^N [\Phi_{s,0}(Z)]](t,x) -  E_i[\tilde \mu^N [\Psi_{s,0}(Z)]](t,x) \bigr\rvert \end{align}
						\noindent for $i=1,2,3$ and similarly for the magnetic field components. Here, we have introduced as an intermediate, the field corresponding to the (smeared) point-charge density $\mu^N [\Phi_{s,0}(Z)]$ of the mean field flow $\Phi_{s,0}(Z)$. We will use a law-of-large number estimate to show that terms of the form \eqref{macroterm0} are typically small, because the particles evolving with the mean field flow are at all times i.i.d. with law $f^N$. For the terms of the form \eqref{microterm0}, we will derive a local Lipschitz bound in terms of $J^N_t(Z)$, the (weighted) maximal distance between the respective mean field and microscopic trajectories.\\
			
\noindent The relevance of $\IE^N_0(J^N_t)$ for the proof of molecular chaos is grounded in the following observations.

\begin{Lemma}\label{Lemma:maxWinfty}
	For $X =(x_1,...,x_n) \in \IR^n$ let $\mu^N[X] := \frac{1}{N} \sum\limits_{i=1}^N \delta_{x_i} \in \M(\IR^n)$. 
	Then we have for all $p \in [1,\infty]$:
	\begin{equation} W_p (\mu^N[X] , \mu^N[Y]) \leq \bigl\lvert X - Y\bigr\rvert_\infty.\end{equation}
\end{Lemma}

\begin{proof}
	Since $W_p \leq W_q$ for $p\leq q$, it suffices to consider the infinite Wasserstein distance defined by	
	\begin{equation*} W_\infty(\mu, \nu) = \inf \lbrace \pi- esssup\,  \lvert x - y \rvert \, \bigr\rvert \, \pi \in \Pi(\mu,\nu)\rbrace.\end{equation*}	
	\noindent  We then observe that $\pi_0 =  \sum\limits_{i=1}^N \delta_{x_{i}}\delta_{y_{i}} \in \Pi(\mu^N[Z], \mu^N[Y])$ with  $ \pi_0- esssup\,  \lvert x - y \rvert  = \max\limits_{1 \leq i \leq N} \lvert x_i - y_i\rvert  = \lvert X - Y \rvert_\infty$. 
\end{proof}
\noindent With this Lemma, we immediately conclude the following:
\begin{Proposition}\label{Prop:mufromJ}
	For all $p \in [1,\infty]$ it holds that 
	\begin{equation}  
	\IP_0 \Bigl[ \sup\limits_{0 \leq s \leq t} W_p(\mu^N[\Psi_{s,0}(Z)], \mu^N[\Phi_{s,0}(Z)]) \geq N^{-\delta} \Bigr] \leq \IE_0(J^N_t).\end{equation}
\end{Proposition}

			\noindent In total, the approximation of the solution to the Vlasov-Maxwell system will be split as:
			\begin{align}\label{split1} W_p(\mu^N_t[Z] , f_t ) &\leq W_p(\mu^N[\Psi_{t,0}(Z)], \mu^N[\Phi_{t,0}(Z)])\\\label{split2}
				&+ W_p(\mu^N[\Phi_{t,0}(Z)], f^N_t)\\\label{split3}
				&+ W_p(f^N_t, f_t).\end{align}
			
			\noindent The first term is the most interesting one, concerning the difference between microscopic time-evolution and mean field time-evolution. It will be controlled in terms of $\IE^N_0(J^N_t)$ by virtue of Prop. \ref{Prop:mufromJ}.\\

			\noindent Convergence of \eqref{split3} is a purely deterministic statement and follows from Theorem \ref{Thm:Rein} cited above. The proof of Rein, however, is based on a compactness argument and does not yield quantitative bounds. Hence, we do not know at what rate $\eqref{split3}$ goes to zero. Based on the corresponding result in the Vlasov-Poisson case, see  \cite{Dustin}, we conjecture that $W_p(f^N_t, f_t) \sim r_N^{1-\epsilon}$ for any $\epsilon > 0$ and $p \leq 2$, though we were not yet able to prove this. \\
			
			\noindent The second term $W_p(\mu^N[\Phi_{t,0}(Z)], f^N_t) = W_p(\varphi^N_{t,0} \# \mu^N_0[Z], \varphi^N_{t,0} \#  f_0)$ concerns the sampling of the mean field dynamics by discrete particle trajectories. Since the mean field forces satisfy a Lipschitz bound uniformly in $N$ according to \eqref{LipschitzL}, we have the following standard result:
			
			\begin{Lemma}\label{Lemma:LGronwall}
				Under the assumptions of Theorem \ref{Thm:Thm3}, it holds that
				\begin{equation*}
					W_p(\mu^N[\Phi_{t,0}(Z)], f^N_t) = W_p(\varphi^N_{t,0} \# \mu^N_0[Z], \varphi^N_{t,0} \#  f_0)\leq e^{tL} W_p(\mu^N_0[Z], f^N_t)
				\end{equation*}
				for all $0 \leq t \leq T$, where $L$ is the uniform Lipschitz constant defined in \eqref{LipschitzL}.
			\end{Lemma}
			
				
			
			\noindent It remains to check that if the initial configuration $Z$ is chosen randomly with law $\otimes^N f_0$, the microscopic density $\mu^N_0[Z]$ approximates $f_0$ in Wasserstein distance. To this end, we will apply the following large deviation estimate due to  Fournier and Gullin \cite{Fournier}. 
			
			\begin{Theorem}[Fournier and Guillin]\label{Fournier}
				\noindent Let $f \in \M(\IR^n)$ and $p \in (0, \infty)$. For  $q > 0, \kappa > 0$, and $\gamma >0$. Assume there exists $ \kappa > $ and $\gamma > 0$ such that $E_{\kappa, \gamma}(f):= \int e^{\gamma \lvert x \rvert^\kappa}  \dd f(x) < + \infty$. 
				Let $(x_i)_{i=1,...,N}$ be a sample of independent variables, distributed according to the law $f$ and $\mu^N[X]:= \sum\limits_{i=1}^N \delta_{x_i}$.
				Then, for all $N \geq 1$ and $\xi \in (0, 1)$: 
				\begin{equation*} 
				\IP\bigl[W^p_p(\mu^N[X], f) > \xi \bigr] \leq  a(N,\xi)
				\end{equation*}
				with \begin{equation*}	
				a(N,\xi):=	
				C \begin{cases}\exp(-cN\xi^2) & \text{if } p > n/2\\
				\exp(-cN(\frac{\xi}{\ln(2+1/\xi)})^2) &  \text{if } p = n/2\\
				\exp(-cN\xi^{n/p}) &\text{if } p \in [1,n/2) \end{cases} 
				\end{equation*}
			
				\noindent The positive constants $C$ and $c$ depend only on $p$, $n$ $\kappa, \gamma$ and $E_{\kappa, \gamma}(f)$.
			\end{Theorem}

			\begin{Lemma}\label{Cor:largedeviation}
				Applying the previous theorem in dimension $n=6$ with $\epsilon = N^{\alpha p}$ we get 
				\begin{equation*}	\IP\Bigl[W_p(\mu^N_0[Z], f_0) > N^{-\alpha} \Bigr] \leq  a(N,p,\alpha) = c' \cdot  \begin{cases}\exp(-cN^{1-2p \alpha}) & \text{if } p > 3\\
						\exp(-c\frac{N^{1- 6 \alpha}}{\log(2+N^{3\alpha})^{2}}) &  \text{if } p = 3\\
						\exp(-cN^{1- 6 \alpha}) &\text{if } p \in [1,3). \end{cases}\end{equation*}
			\end{Lemma}

			\section{Global estimates}\label{Section:globalestimates}
			By assumption, there exists a constant $C_0 >0$ such that $\lVert \rho[f^N] \rVert_{L^\infty([0,T] \times \IR^3)} \leq C_0$ for all $N \in \IN \cup \lbrace + \infty \rbrace$. Using the methods introduced in \cite{Dustin}, we will now show that as long as mean field dynamics and microscopic dynamics are sufficiently close, this implies certain bounds on the microscopic density and fields. As we have to deal with singular kernels, the necessary regularizations come from the smearing with the $N$-dependent mollifier $\chi^N$.\\
			
			\noindent \textbf{Notation / Definition:} Following \cite{Pallard} we introduce the shorthand notation
			\begin{equation} g \lesssim h :\iff \exists C > 0: g \leq C \, h,\end{equation}
			where $C \in \IR$ is a constant that may depend only on $T$ and initial data.\\ 
			Moreover, for fixed $N \geq 1$ and any measurable function $h$ on $\IR^n$, $n= 3$ or $n=6$, we introduce the notation $\tilde h := \chi^N*_x h$. For a probability measure $\M(\IR^n)$ we define $\tilde v \in \M(\IR^n)$ by $\int h \dd\tilde{\nu} := \int \tilde h \dd\nu$ for all measureable $h$.  Note that if $\rho(x) = \frac{1}{N} \sum \limits_{i=1}^N \delta(x - x_i)$ for $x_i \in \IR^3$, we have $\tilde \rho = \frac{1}{N} \sum \limits_{i=1}^N \chi^N(x - x_i)$, consistent with the notation of Section \ref{Section:microscopictheory}.

			\begin{Lemma}\label{chibounds}
				Let $h: \IR^3 \to \IR^n$ a measurable function satisfying $\lvert h(x) \rvert \leq \frac{1}{\lvert x \rvert^2}$. Then:
				\begin{align}
					i)	\;\; \; &\lvert \chi^N*h (x) \rvert \;\;\; \lesssim \min \bigl\lbrace r_N^{-2}, \frac{1}{\lvert x \rvert^2} \bigr\rbrace,\\[1ex]
					ii) \;\; \; &\lvert \nabla \chi^N*h (x) \rvert \lesssim \min \bigl\lbrace r_N^{-3}, \frac{1}{\lvert x \rvert^3} \bigr\rbrace.
				\end{align}
			\end{Lemma}
			\begin{proof}
				Recalling that $\lVert \chi^N \rVert_\infty = r_N^{-3} \lVert \chi \rVert_\infty$ and $\lVert \chi^N \rVert_1 = 1$, we compute:
				\begin{align*}	
					\lvert \chi^N * h (x)\rvert &\leq \int \lvert k (y) \rvert \chi^N(x-y) \dd^3 y \leq \int \frac{1}{\lvert y \rvert^2}\, \chi^N(x-y) \, \dd^3 y\\
					&\leq \int\limits_{\lvert y \rvert \leq r_N} +  \int\limits_{\lvert y \rvert > r_N} \frac{1}{\lvert y \rvert^2} \chi^N(x-y) \dd^3 y\\
					&\leq \lVert \chi^N \rVert_\infty  \int\limits_{\lvert y \rvert \leq  r_N} \frac{1}{\lvert y \rvert^2} \dd^3 y +  \frac{1}{r_N^{2}} \int \chi^N (x-y) \dd^3 y \lesssim  r_N^{-2}.
				\end{align*}
				
				\noindent Similarly,
				\begin{align*}
					\lvert \nabla (\chi^N*h) (x) \rvert &\leq \lvert \nabla \chi^N \rvert * \lvert k \rvert (x) \leq  \int\limits_{\lvert y \rvert \leq r_N} +  \int\limits_{\lvert y \rvert > r_N} \frac{1}{\lvert y \rvert^2} \vert \nabla \chi^N(x-y) \rvert \dd^3 y\\
					&\leq \lVert \nabla \chi^N \rVert_\infty  \int\limits_{\lvert y \rvert \leq  r_N} \frac{1}{\lvert y \rvert^2} \dd^3 y +  \frac{1}{r_N^{2}} \int \lvert \nabla \chi^N (x-y) \rvert \dd^3 y\\
					& \leq r_N^{-4} \lVert \nabla \chi \rVert_\infty \, 4 \pi  r_N  + {r_N^{-2}} r_N^{-1} \lVert \nabla \chi \rVert_1
					\leq r_N^{-3} (4 \pi  \lVert \nabla \chi \rVert_\infty + \lVert \nabla \chi \rVert_1).
				\end{align*}
				
				\noindent Finally, if $\lvert x \rvert > 2r_N$, the mean-value theorem of integration yields for $s \geq 1$:
				\begin{align*} \chi^N*\frac{1}{\lvert y \rvert^s}(x) = \int \frac{1}{\lvert x-y \rvert^s} \chi^N(y) \dd ^3 y &\leq \sup\lbrace  \lvert x - y \rvert^{-s} \mid y \in \supp \chi^N \rbrace\leq \frac{2^s}{\lvert x \rvert^s}, \end{align*} 
				where we used the fact that $\int \chi^N =1$ and $\lvert y \rvert \leq r_N \leq \frac{1}{2} \lvert x \rvert, \, \forall y \in \supp (\chi^N)$. 
			\end{proof}

			\subsection{Bounds on the charge density}
			\begin{Proposition}\label{Prop:rhobound}
				
				Suppose there exists a $p \in [1, \infty)$ such that 
				
				\begin{equation}\label{sufficientlyfast2} W_p(\mu^N_0[Z], f_0) \leq r_N^{3 + p}. \end{equation}
				
				\noindent Then there exists a constant $C_\rho$ depending on $T$  such that
				\begin{equation} 
					\lvert {}^{N}\Psi_{t,0}(Z) - {}^{N}\Phi_{t,0}(Z) \rvert_\infty < r_N \Rightarrow \lVert \tilde \rho[\mu^N_t[Z]] \rVert_\infty \leq C_\rho.
				\end{equation}
			\end{Proposition}
			
			\begin{Corollary}
				Under the conditions of the proposition, we also have
				\begin{equation}\label{CorDbound}
					\lvert {}^{N}\Psi_{t,0}(Z) - {}^{N}\Phi_{t,0}(Z) \rvert_\infty < r_N \Rightarrow 	\lVert D^\alpha  \tilde\rho[\mu^N_t[Z]]  \rVert_\infty \lesssim r_N^{-\lvert \alpha \rvert}.
				\end{equation}
			\end{Corollary}
			\begin{proof}
				Note that $  D^\alpha  \tilde\rho[\mu^N_t] = D^\alpha (\chi^N * \rho[\mu^N_t]) =   (D^\alpha \chi^N) * \rho[\mu^N_t]$, and  \begin{equation*} D^\alpha \chi^N(x)  =  D^\alpha_x  r_N^{-3}\chi(\frac{x}{r_N}) = r_N^{-\lvert \alpha \rvert} r_N^{-3} (D^\alpha\chi)(\frac{x}{r_N}). \end{equation*}
				\noindent Let $\overline \chi := \frac{D^\alpha\chi}{\lVert D^\alpha\chi \rVert_1}$. This $\overline \chi$ satisfies \eqref{formfactor} and can thus be used as a form factor instead of $\chi$. The previous proposition then yields
				$
				\lvert {}^{N}\Psi_{t,0}(Z) - {}^{N}\Phi_{t,0}(Z) \rvert_\infty < r_N \Rightarrow
				\lVert \overline\chi^N * \rho[\mu^N_t]\rVert_\infty \leq C,
				$
				and thus 
				\begin{equation*}	\lVert D^\alpha  \tilde\rho[\mu^N_t] \rVert_\infty  = \lVert D^\alpha\chi \rVert_1 \,  r_N^{- \lvert \alpha \rvert} \lVert \overline\chi^N * \rho[\mu^N_t]\rVert_\infty \lesssim  r_N^{- \lvert \alpha \rvert}.
				\end{equation*}
			\end{proof}
			
			\begin{Remark} In the end, we will have to show that assumption \eqref{sufficientlyfast2} is satisfied for \emph{typical} initial conditions, as the initial particle configurations are chosen randomly and independently with law $f_0$. This (and only this) requirement will set the lower bound on the cut-off to $r_N \sim N^{-\gamma}$ with $\gamma < \frac{1}{12}$. 
			\end{Remark}
			
			\noindent The proof of Proposition \ref{Prop:rhobound} is based on the following Lemma derived in \cite{Dustin} (c.f. also \cite[Prop. 2.1]{BGV}.) 
			
	\begin{Lemma}\label{Lemma:rhobound} 
		Let $\rho_1, \rho_2$ two probability measures on $\IR^d$ and $\rho_2 \in L^\infty(\IR^d)$. Then:
		\begin{equation} 
		\lVert \tilde{\rho}_1  \rVert_\infty \leq \lvert \mathrm{B}^d(2) \rvert\, \lVert \rho_2 \rVert_\infty +  r_N^{-(p+d)} \,  W_p^p(\rho_1, \rho_2),
		\end{equation}
		where $\mathrm{B}^d(2) \subset \IR^d$ is the $d$-dimensional ball with radius 2. 
	\end{Lemma}
	
	\begin{proof}
		For any integrable function $\Phi$, we consider the $c$-conjugate
		\begin{equation*} \Phi^c(y) := \sup\limits_{x} \lbrace \Phi(x) - \lvert x-y \rvert^p \rbrace\end{equation*} 
		This is the smallest function satisfying $\Phi^c(y) \geq \Phi(y)$ and $\Phi(x) - \Phi^c(y) \leq \lvert x-y \rvert^p, \, \forall x,y \in \IR^d$.\\ Now, we write
		\begin{equation*} \begin{split}\tilde{\rho}_1 (x) = r_N^{-(d+p)} \Bigl[\int r_N^{d+p}\chi^N(x-y)  \rho_1(y) \dd y - \int (r_N^{d+p}\chi^N(x-\cdot))^c(z) \rho_1(z) \, \dd z\\ + \int (r_N^{d+p}\chi^N(x-\cdot))^c(z) \, \rho_1(z) \dd z \Bigr] \end{split} \end{equation*}
		By the Kantorovich duality theorem \eqref{Kantorovich} we have
		\begin{equation*} \int r_N^{d+p} \chi^N(x-y)\,  \rho_1(y) \dd y \, - \int (r_N^{d+p}\chi^N(x-\cdot))^c(z)\,  \rho_2(z) \dd z \leq W_p^p(\rho_1, \rho_2). \end{equation*}
		It remains to estimate  
		\begin{equation*} \int (r_N^{d+p} \chi^N(x-\cdot))^c(z)\,  \rho_2(z) \, \dd z. \end{equation*}
		
		\noindent Recalling that $\lVert \chi^N \rVert_\infty = r_N^{-d}$, we find 
		\begin{equation*}(r_N^{d+p} \chi^N(x-\cdot))^c(z) = \sup\limits_{y \in \IR^3} \lbrace  r_N^{d+p} \chi^N(x-y) - \lvert y-z \rvert^p \rbrace \leq r_N^{d+p} \lVert \chi^N \rVert_\infty = r_N^{p}.\end{equation*} 
		Moreover, we observe that 
		\begin{equation} \supp (r_N^{d+p} \chi^N(x-\cdot))^c \subseteq \mathrm{B}(2r_N ; x) := \lbrace z \in \IR^3 : \lvert z -x \rvert \leq 2r_N \rbrace, \end{equation}
		since $ \lvert z- x \rvert >2r_N$ implies $\chi^N(x-y) = 0$, unless $\lvert y-z \rvert \geq r_N$. But then: $r_N^{d+p} \chi^N(x-y) - \lvert y-z \rvert^p \leq r_N^{d+p} r_N^{-d} - r_N^p = 0$. 
		Hence,
		\begin{equation*} \begin{split} \int (r_N^{d+p}\chi^N(x-\cdot))^c(z)  \rho_2(z) \dd z
		\leq \lVert  \rho_2 \rVert_\infty \,r_N^p \, \lvert \mathrm{B}(2 r_N; x) \rvert
		\leq 2^d \lvert \mathrm{B}^d(1)\rvert \, \lVert  \rho_2 \rVert_\infty\, r_N^{d+p}. \end{split}\end{equation*}
		
		\noindent In total, we find
		\begin{equation*} \lVert \tilde{\rho}_1 \rVert_\infty \leq r_N^{-(p+d)} \, W_p^p(\rho_1, \rho_2) + \lvert B^d(2)\rvert \lVert  \rho_2 \rVert_\infty\end{equation*}
		as announced.
	\end{proof}
			
			\begin{proof}[\textbf{\em{Proof of Proposition \ref{Prop:rhobound}}}] 
				\noindent As an intermediate step, we introduce the density $\mu^N[\Phi_{t,0}(Z)]$ corresponding to the mean field flow defined in \ref{Def:macroflow}. Since the mean field force is Lipschitz continuous with a constant $L$ independent of $N$, we have according to Lemma \ref{Lemma:LGronwall}
				\begin{equation*} W^p_p(\mu^N[\Phi_{t,0}(Z)], f^N_t) \leq e^{tL}  W^p_p(\mu^N_0[Z], f_0). \end{equation*}
				Moreover, by assumption, $\lVert \tilde \rho[f^N_t] \rVert_\infty \leq \lVert \rho[f^N_t] \rVert_\infty \leq C_0, \, \forall N$. Applying the previous Lemma with $\rho_1 = \rho[\mu^N[\Phi_{t,0}(Z)]],  \, \rho_2  =  \rho[f^N_t]$, we get
				\begin{equation*}
					\lVert \tilde\rho[\mu^N[\Phi_{t,0}(Z)]] \rVert_\infty \lesssim C_0 + e^{tL}.
				\end{equation*}
				
				\noindent  Now, recall from Lemma  \ref{Lemma:maxWinfty} that $W_\infty (\mu[\Phi_{t,0}(Z)], \mu[\Psi_{t,0}(Z)]) \leq \bigl \lvert \Phi_{t,0}(Z) - \Psi_{t,0}(Z) \bigr \rvert_\infty$, where $W_\infty$ is the infinity Wasserstein distance. If $\bigl \lvert \Phi_{t,0}(Z) - \Psi_{t,0}(Z) \bigr \rvert_\infty < r_N$, there exists $q > 0$ such that  $\bigl \lvert \Phi_{t,0}(Z) - \Psi_{t,0}(Z) \bigr \rvert_\infty \leq  r_N^{1+\frac{3}{q}}$. We thus have 
				\begin{align*} r_N^{-(q+3)} W_q^q(\mu^N[\Phi_{t,0}(Z)], \mu^N[\Psi_{t,0}(Z)]) &\leq r_N^{-(q+3)} (W_\infty (\mu[\Phi_{t,0}(Z)], \mu[\Psi_{t,0}(Z)]))^q\\
					&\leq  r_N^{-(q+3)} \bigl \lvert \Phi_{t,0}(Z) - \Psi_{t,0}(Z) \bigr \rvert^q_\infty \leq 1. \end{align*}
				Applying once more Lemma \ref{Lemma:rhobound} with  $\rho_1 = \rho[\mu^N[\Psi_{t,0}(Z)]],  \, \rho_2  = \rho[\mu^N[\Phi_{t,0}(Z)]]$ and the Wasserstein metric of order $q$, we get the announced result. 
				
			\end{proof}

			\subsection{Bounds on the field derivatives}
			\begin{Proposition}\label{Prop:deribound}
				Under the conditions of Proposition \ref{Prop:rhobound}, the microscopic fields satisfy
				\begin{align} \lVert \nabla_x  E_t[\tilde \mu^N] \rVert_\infty, \,    \lVert \nabla_x  B_t[\tilde \mu^N]\rVert_\infty \; \lesssim r_N^{-2}. \end{align}
			\end{Proposition}
			\begin{proof}
				
				\noindent We begin with the homogeneous field
				\begin{equation} E_0(t,x) = \partial_t Y(t, \cdot) *E_{in} (x) = \partial _t \Bigl(\frac{t}{4\pi}\int\limits_{S^2} E_{in}( y + \omega t) \dd \omega \Bigr). \end{equation}
				\noindent From this representation, one reads of the bounds
				\begin{equation} \lVert E_0(t, \cdot) \rVert_{W^{k-1,\infty}_x} \leq \lVert E_{in} \rVert_{W^{k-1,\infty}_x} + t \lVert E_{in} \rVert_{W^{k,\infty}_x}.
				\end{equation}
				In particular, for $E_{in} = - \nabla G * \rho_0$, we have
				\begin{align*}
					\lVert D^\alpha E_{in}(t, \cdot) \lVert_\infty & \lesssim \lVert  D^\alpha \rho_0 \rVert_\infty +  \lVert  D^\alpha \rho_0 \rVert_1, \, \lvert \alpha \rvert = 0,1,2,
				\end{align*}
				where we used 
				\begin{align*} \int \frac{1}{\lvert y \rvert^2} \lvert D^\alpha \rho_0\rvert (x-y) \dd^3y &= \int\limits_{\lvert y \rvert \leq 1} + \int\limits_{\lvert y \rvert > 1}  \frac{1}{\lvert y \rvert^2}  \lvert D^\alpha \rho_0\rvert(x-y) \dd^3y\\
					&\leq 4\pi \lVert   D^\alpha \rho_0 \rVert_\infty + \lVert D^\alpha \rho_0 \rVert_1.
				\end{align*}
				\noindent For the inhomogeneous parts, we can use equation \eqref{lorentzforce} to write
				\begin{align*}
					E(t,x) &= -  \int (\nabla_x + v(\eta) \partial_t ) Y*f(\cdot,\cdot,\eta) \dd \eta\\
					&= - \int  (\nabla_x + v(\eta) \partial_t ) \int\limits_{0}^t\int\limits_{S^2} (t-s) f(s, x + \omega(t-s), \eta) \dd \eta,\\ 
					B(t,x) & = -  \int (v(\eta) \times \nabla_x)  Y*f(\cdot,\cdot,\eta) \dd \eta\\
					&= - \int (v(\eta) \times \nabla_x)  \int\limits_{0}^t\int\limits_{S^2} (t-s) f(s, x + \omega(t-s), \eta) \dd \eta,
				\end{align*}
				
				\noindent from which we read off the bounds
				\begin{align}\lVert \nabla E\rVert_\infty, \, \lVert \nabla B\rVert_\infty \leq 4 \pi  (1+T)T  \sup\limits_{s\leq T} \sum\limits_{\lvert \alpha \rvert \leq 2} \lVert D^\alpha \rho[f(s)] \rVert_\infty. \end{align}
				
				\noindent Applying this to $f(t) = \tilde \mu^N_t = \chi^N*_x\mu^N_t[Z]$ and using \eqref{CorDbound}, the desired statement follows. 
			\end{proof}
			
			\subsection{Bound on the total force}
			While we will show that for typical initial conditions, the microscopic time-evolution will be close to the mean field time-evolution, we also need to control how ``bad'' initial conditions contribute to the growth of $\IE_0(J_t)$. To this end, we require a bound on the total microscopic force, although a rather coarse one will suffice. 
			\begin{Proposition}
				The total microscopic force is bounded as
				\begin{equation}
					\lVert \tilde K_t[\tilde \mu^N]\rVert_{L^\infty(\IR^3\times \IR^3)}  \leq \lVert \tilde E_t[\tilde \mu^N] \rVert_{L^\infty(\IR^3)} + \lVert \tilde B_t[\tilde \mu^N] \rVert_{L^\infty(\IR^3)} \lesssim r_N^{-2}, \; \forall t \geq 0. 
				\end{equation}
				Note that this holds independently of assumption \eqref{sufficientlyfast2}. 
			\end{Proposition}
			
			\begin{proof}
				Recall that the total energy
				\begin{equation*}  \varepsilon(t) = \frac{1}{N} \sum\limits_{i=1}^N \sqrt{1 + \lvert \xi_i(t) \rvert^2 } + \frac{1}{2} \int E_t^2(x) + B_t^2(x) \dd x \end{equation*} 
				is a constant of motion.  At $t=0$, we thus have: 
				
				\begin{equation*} \varepsilon(0) \leq  \frac{1}{2}\bigl( \lVert E_{in} \lVert_2^2 + \lVert B_{in} \lVert_2^2 \bigr) + \sqrt{1 + \overline \xi ^2 }. \end{equation*}
				For the microscopic system, we have according to our convention, equation \eqref{microin},
				
				\begin{equation*}E^\mu_{in} :=  E^N_{in} - \nabla G *( \tilde\rho[\mu^N_0[Z]] - \tilde\rho[f_0]), \;\;\;  B^\mu_{in} := B^N_{in}. \end{equation*}
				Since $E^N_{in} = \chi^N*E_{in}$, we have $\lVert E^N_{in} \rVert_2 \leq \lVert E_{in} \rVert_2$ uniformly in $N$. The same holds for $B^\mu_{in} = B^N_{in}.$ It remains to estimate $\lVert \nabla G * \tilde\rho[\mu^N_0[Z]] \rVert_2$ and $\lVert \nabla G * \tilde\rho[f_0] \rVert_2$.

				\noindent Since $\lvert \nabla G(x) \rvert = \frac{1}{4\pi \lvert x \rvert^2}$, Lemma \ref{chibounds} yields $\lvert \chi^N *_x \nabla G \rvert \lesssim \min \lbrace r^{-2}_N, \lvert x \rvert^{-2} \rbrace$ and we compute
				\begin{equation}\begin{split}  \lVert \chi^N * \nabla G \rVert^2_2 &\leq  \int\limits_{\lvert y \rvert \leq r_N} \lvert \chi^N *_x \nabla G \rvert^2 (x) + \int\limits_{\lvert y \rvert > r_N} \lvert \chi^N *_x \nabla G \rvert^2(x)\\
						&\lesssim   r_N^{-4}  \int\limits_{\lvert x \rvert < r_N} \dd ^3 x + \int\limits_{\lvert x \rvert \geq r_N} \lvert x \rvert^{-4} \dd ^3 x\\
						& \lesssim r_N^{-4}\, r_N^3 + r_N^{-1} = 2 r_N^{-1}.
					\end{split}\end{equation} 
					This yields, on the one hand, 
					\begin{equation} 
						\lVert \nabla G * \tilde\rho[\mu^N_0[Z]] \rVert^2_2 = \bigl \lVert \frac{1}{N} \sum\limits_{i=1}^N \nabla G * \chi^N(\cdot - x_i(0)) \bigr\rVert^2_2 \leq \lVert\chi^N * \nabla G \rVert_2^2 \lesssim r_N^{-1},\end{equation}
					and, on the other hand, 
					\begin{equation}
						\lVert \nabla G * \tilde\rho[f_0] \rVert_2 = \lVert \chi^N*\nabla G *\rho[f_0] \rVert_2 \leq \lVert \chi^N*\nabla G  \rVert_2 \lVert \rho[f_0] \rVert_1 \lesssim r_N^{-1/2}. \end{equation}
					
					\noindent In total, we have found that
					\begin{equation} 
						\lVert E(t,\cdot) \rVert_2 + \lVert B(t,\cdot) \rVert_2 \leq \sqrt{2 \varepsilon + 1 + \overline{\xi}^2} \lesssim r_N^{-1/2}.
					\end{equation}

					\noindent Finally, by Young's inequality, we have for $\tilde K(t,x,\xi) = \chi^N*_x(E_t + v(\xi)\times B_t) (t,x)$:
					\begin{equation*}\begin{split} \lVert \widetilde K[\tilde \mu^N](t,\cdot,\cdot) \rVert_\infty \leq \lVert \chi^N \rVert_2 \bigl( \lVert E[\tilde \mu^N](t, \cdot) \rVert_2 +  \lVert B[\tilde \mu^N](t, \cdot) \rVert_2  \bigr)
							\lesssim r_N^{-3/2} r_N^{-1/2} = r_N^{-2},\end{split}\end{equation*}
					where we used 
					\begin{align*} 
						&\lVert \chi^N \rVert^2_2 = \int (\chi^N(x))^2 \dd^3 x = \int (r_N^{-3}\chi({x}/{r_N}))^2  \dd^3 x 
						= r_N^{-3} \int \chi(y)^2  \dd^3 y = r_N^{-3} \lVert \chi \rVert_2^2. \end{align*}
				\end{proof}
				
				\noindent It might be interesting to note that -- in contrast to the other mean field results presented or referenced in this thesis -- we actually use an energy bound here, exploiting the conservation of energy in the Abraham model. Also note that this is the only bound for which we have to use both mollifiers appearing in \eqref{RMV}.

				\section{Light cone structure}\label{section:lightconestructure}
				
				The Maxwell theory as well as the Vlasov-Maxwell approximation are relativistic. Particle interactions -- mediated by the electromagnetic field -- are retarded, with influences ``propagating'' with the speed of light. More precisely, the field value at a given space-time point $(t,x) \in \IR \times \IR^3$ depends on the particle trajectories only at their intersection with the backwards light cone $\lbrace (s,y) \mid (t-s)^2 - (x-y)^2 = 0, t-s \geq 0 \rbrace$. Formally, this light cone structure is manifested in the d'Alembert kernel $Y(t,x)$ defined in \eqref{dAlembert}, which has support in $\lbrace t = \lvert x \rvert, t >0 \rbrace$. The regularized Vlasov-Maxwell system \eqref{RMV} is only semi-relativistic (because of the rigid form factor), but inherits this light-cone structure. Integral expressions of the form (\ref{E1term}, \ref{E2term}), determining the inhomogeneous field components, evaluate the mean field density on the backwards light cone. Since the Vlasov density is transported with the characteristic flow, the respective integrals can be pulled-back to the $t=0$ hypersuface in a canonical way. The respective field components at a space-time point $(t,x)$ then depend on the initial distribution $f_0$ on $\mathrm{B}_t(x) \times \IR^3$ where $\mathrm{B}_t(x) = \mathrm{B}(t;x)$ is the ball around $x$ with radius $t$. In the following, we make these observations more precise.

				\begin{Definition}[Retarded time]
					Fix a spacetime point $(t,x) \in \IR \times \IR^3$. Let $f_t$ a solution of \eqref{RMV} and $\varphi_{s,0}(z)= (y^*(s,z), \eta^*(s,z))$ the characteristic flow, i.e. the solution of \eqref{meanfieldeq} with $(y^*(0), \eta^*(0)) = z$. Then we denote by $t_{ret}(z)$ the unique solution of  
					\begin{equation} (t-s)^2 - (x - y^*(s, z))^2 = 0; \; (t-s) > 0. \end{equation}
					
					\noindent $t_{ret}(z) = t_{ret}(y^*(s,z); t,x)$ is the time at which the trajectory $y^*(s)$ crosses the backward light cone with origin $(t,x)$. We have $t_{ret}(z) \geq 0 \iff y_0 \in \mathrm{B}_t(x)= \lbrace y \in \IR^3 :  \lvert x -y \rvert \leq t \rbrace$. 
				\end{Definition}
				
				\begin{Lemma}[Distributions on the light cone]\label{Lemma:lightconedistribution}
					
					Let $f_t$ a solution of \eqref{RMV} and $\varphi_{s,0}(z)= (y^*(s,z), \eta^*(s,z))$ as above. For a fixed space-time point $(t,x) \in \IR^+ \times \IR^3$ consider the diffeomorphism
					\begin{equation}\label{LCdiffeo}\begin{split} \phi: \mathrm{B}_t(x) \times \IR^3 &\to \mathrm{B}_t(x) \times \IR^3\\
							z= (x,\xi) &\mapsto (y^*(t_{ret} (z),z), \eta^*(t_{ret} (z),z)).
						\end{split}  \end{equation}
						
						\begin{enumerate}[1)]
							\item For $a \in C(\IR^3 \times \IR^3)$, we have (with  $n(x - y) = \frac{x - y}{\lvert x - y \rvert}$):
							\begin{equation}\label{LCditribution1}\begin{split} 
									\int\limits_{\mathrm{B}_t(x) \times \IR^3}& a (\phi(z)) \; f_0(z) \, \dd z \\
									& = \int\limits_{\mathrm{B}_t(x) \times \IR^3} a (y, \eta) \, (1 - n(x-y) v(\eta) ) \, f(t - \lvert x - y\rvert , y, \eta)  \dd y \,  \dd \eta.
								\end{split}\end{equation}
								
								\item For  $\alpha \in C(\IR\times\IR^3 \times \IR^3)$:
								\begin{multline}\label{LCdistribution2}
									\int\limits (\alpha Y) *_{t,x} (\mathds{1}_{t\geq 0} f) (t,x, \eta)\,\dd \eta\, \\
									= \int\limits_{\mathrm{B}_t(x) \times \IR^3}\frac{\alpha(t-s, x -y^*(s,z), \eta^*(s,z))}{\lvert x - y^*(s,z) \rvert (1 -  n(x-y^*(s,z)) \cdot v(\eta^*(s,z)) } \Biggl \lvert_{s = tret(z)}  f_0(z) \, \dd z.
								\end{multline}
							\end{enumerate}
							
						\end{Lemma}
						
						\begin{proof}
							Since $f_t=\varphi_{t,0} \# f_0$, we compute
							\begin{align*}  &\int\limits_{ \mathrm{B}_t(x) \times \IR^3} a (y, \eta) \, f(t - \lvert x - y\rvert , y, \eta)  \dd y \,  \dd \eta\\
								& = \int_{[0,t] \times \mathrm{B}_t(x) \times \IR^3} a(y, \eta) \, \delta(\lvert x - y\rvert - (t-s)) \, f(s,y,\eta) \dd s \dd y \dd \eta\\
								& = \int a(y, \eta) \, \delta(\lvert x - y\rvert - (t-s)) \,\varphi_{s,0}\# f_0(y,\eta) \dd s \dd y \dd \eta\\
								& =  \int a(y^*(s; y,\eta), \eta^*(s; y,\eta)) \, \delta( \lvert x - y^*(s; y,\eta) \rvert - (t-s)) \, f_0(y,\eta) \dd s \dd y \dd \eta.
							\end{align*}
							
							\noindent Now we use: If $h \in C^1$ has a unique root $\zeta$, then $\delta(h(x)) = \delta(x - \zeta){h'(\zeta)}$ in the sense of distributions. The function $h(s) =  \lvert x - y^*(s; y,\eta) \rvert - (t-s)$ is differentiable with $h'(s)= 1 - \frac{(x - y^*(s)) \cdot v(\eta^*(s))}{\lvert x - y^*(s) \rvert } =  1 - n(x - y^*(s)) \cdot v(\eta^*(s)) $. If $y^*(0) \in \mathrm{B}_t(x)$, it has a unique positive root $t_{ret} = t_{ret}(z)$. Hence, we get:
							\begin{equation}\begin{split}
									\int a(y, \eta) \, \delta(t - s - &\lvert x - y \rvert) \, f(s,y,\eta)\,  \dd s \dd y \dd \eta\\
									&=  \int \frac{a(y^*( t_{ret}(z), z), \eta^*( t_{ret}(z), z) )}{1 -  n(x-y^*( t_{ret}(z))) \cdot v(\eta^*( t_{ret}(z))) }\; f_0(z)\, \dd z
								\end{split}\end{equation}
								and the identity follows. For \eqref{LCdistribution2}, we have
								\begin{align*}
									&\int\limits (\alpha Y) *_{t,x} (\mathds{1}_{t\geq 0} f) \dd \eta (t,x) \\
									&= \int\limits_{\IR \times \IR^3 \times \IR^3} \alpha (t-s, x -y , \eta) Y(\lvert x - y \rvert - (t-s)) \mathds{1}_{\lbrace s \geq 0\rbrace} f(s, y, \eta) \dd s \dd y \dd \eta.
								\end{align*}
								Now observe that on the support of $Y$, we have $\mathds{1}_{\lbrace s \geq 0\rbrace} = \mathds{1}_{\lbrace y \in B_t(x) \rbrace}$ and $(t-s) = \lvert x - y \rvert$ and apply part 1) of the Lemma to $a(y,\eta) =  \lvert x - y\rvert^{-1} \alpha (\lvert x -y \rvert, x -y , \eta)$. 
							\end{proof}
							
							\noindent Furthermore, in order to compare the fields generated by the mean field trajectories with those generated by the microscopic trajectories, we will require the following lemma. 
							
							\begin{Lemma}\label{Lemma:closetret}
								Let $x^*_1(s), x_2^*(s)$ two trajectories with velocity bounded by $\overline{v}< 1$. Fix a space-time point $(t,x) \in \IR \times \IR^3$ and denote by $t^i_{ret}, \, i=1,2$ the time at which trajectory $i$ intersects the backward light cone with origin $(t,x)$. Then we have:
								
								\begin{equation} \lvert x^*_1(t^1_{ret}) -   x^*_2(t^2_{ret}) \rvert \leq \frac{1}{1 - \overline{v}} \;   \lvert x^*_1(t^1_{ret}) -   x^*_2(t^1_{ret}) \rvert. \end{equation}
								Similarly, if we denote that respective momenta by $\xi_1(s), \xi_2(s)$ and assume that the force $\dot{\xi}_2$ is bounded by $L < \infty$, then
								\begin{equation} \lvert \xi^*_1(t^1_{ret}) -   \xi^*_2(t^2_{ret}) \rvert \leq \lvert \xi^*_1(t^1_{ret}) -   \xi^*_2(t^1_{ret}) \rvert  + \frac{L}{1 - \overline{v}} \;   \lvert x^*_1(t^1_{ret}) -   x^*_2(t^1_{ret}) \rvert. \end{equation}
								
							\end{Lemma}
							
							\begin{proof}
								Suppose w.l.o.g. that 
								\begin{align*} (t-t_{ret}^1) - \lvert x - x_1^*(t_{ret}^1) \rvert = 0, \\
									(t-t_{ret}^1) - \lvert x - x_2^*(t_{ret}^1) \rvert > 0.
								\end{align*}
								
								\noindent Set $r:= \lvert x_1^*(t_{ret}^1) - x_2^*(t_{ret}^1) \rvert $ and $\tau = \min \lbrace t, t_{ret}^1 + \frac{r}{1 -\overline {v}} \rbrace$. Obviously, if $\tau =t$, we have
								\begin{equation*}
									(t -\tau) - \lvert x - x_2^*(\tau) \rvert  = - \lvert x - x_2^*(\tau) \rvert \leq 0.
								\end{equation*}
								If $\tau = t_{ret}^1 + \frac{r}{1 -\overline {v}} < t$, we estimate
								\begin{align*}
									\lvert x - x_2^*(\tau) \rvert &\geq \lvert x - x^*_1(t_{ret}^1) \rvert - \lvert x_1^*(t_{ret}^1) - x_2^*(t_{ret}^1) \rvert - \lvert x_2^*(t_{ret}^1) - x_2^*(s_2) \rvert\\
									& \geq (t -t_{ret}^1) - r - \overline{v} (\tau - t_{ret}^1)\\
									& = (t - \tau) + (\tau - t_{ret}^1) - r - \overline{v} (\tau - t_{ret}^1)\\
									& = (t - \tau) + (1-\overline{v})(\tau - t_{ret}^1) - r 
								\end{align*}
								and therefore also 
								\begin{equation*}
									(t -\tau) - \lvert x - x_2^*(\tau) \rvert  \leq  r - (1-\overline{v})(\tau - t_{ret}^1) = 0.
								\end{equation*}
								\noindent By continuity, there thus exists $s \in (t_{ret}^1, \tau]$ with $(t - s) - \lvert x - x_2^*(s) \rvert = 0$. Hence, $s = t_{ret}^2$ and we found
								\begin{align*} 
									\lvert x_2^*(t_{ret}^2) - x_1^*(t_{ret}^1) \rvert &\leq \lvert x_2^*(t_{ret}^1) - x_1^*(t_{ret}^1) \rvert + \lvert x_2^*(t_{ret}^2) - x_2^*(t_{ret}^1) \rvert  \\
									&\leq r + \overline{v}(t_{ret}^2 - t_{ret}^1) \leq \frac{r}{1- \overline{v}} =  \frac{ \lvert x_2^*(t_{ret}^1) - x^*(t_{ret}^1) \rvert }{1- \overline{v}},
								\end{align*}
								as well as
								\begin{align*} 
									\lvert \xi_2^*(t_{ret}^2) - \xi_1^*(t_{ret}^1) \rvert &\leq \lvert \xi_2^*(t_{ret}^1) - \xi_1^*(t_{ret}^1) \rvert + \lvert \xi_2^*(t_{ret}^2) - \xi_2^*(t_{ret}^1) \rvert  \\
									&\leq \lvert \xi_2^*(t_{ret}^1) - \xi_1^*(t_{ret}^1) \rvert + L \lvert t_{ret}^2 - t_{ret}^1 \rvert\\
									&\leq \lvert \xi_2^*(t_{ret}^1) - \xi_1^*(t_{ret}^1) \rvert + \frac{L}{1-\overline{v}} \lvert x^*_1(t^1_{ret}) -   x^*_2(t^1_{ret}) \rvert.
								\end{align*}
							\end{proof}
							
							
							\subsection{Law of large numbers}
							
							Part of our proof consists in sampling the mean field dynamics along (random) trajectories, i.e. approximating the mean field distribution $f^N_t$ with the discrete measure $\mu^N[\Phi_{t,0}(Z)]$, where $\Phi_{t,0}$ is the mean field flow defined in \eqref{Def:macroflow} and $Z \in \IR^{6N}$ is random with distribution $\otimes^N f_0$. One advantage of this approach is that the $N$ particles evolving with the mean field flow remain i.i.d. with law $f^N_t$ for all times, thus allowing for law of large numbers estimates. We will work with the following (more or less standard) result:
							
							\begin{Proposition}\label{Prop:LLN} Let $f_0 \in L^1\cap L^\infty(\IR^3 \times \IR^3)$ a probability density. Let $\alpha, \beta > 0$ with  $\alpha + \beta < \frac{1}{2}$. Let $h:\IR^6 \to \IR$ such that $\lvert h(z) \rvert \lesssim N^\alpha$. Let $\phi: \IR^6 \to \IR^6$ a diffeomorphism with bounded derivative. Then, for all $\gamma > 0$  there exists a $C_\gamma >0$ such that
								\begin{align} \IP_0\Bigl[\Bigl\lvert \frac{1}{N}\, \sum\limits_{i = 1}^N h(\phi(z_i)) -  \int h(\phi(z)) f_0(z) \Bigr\rvert \geq N^{-\beta}\Bigr] \leq \frac{C_\gamma}{N^\gamma}.
								\end{align}	
								
							\end{Proposition}
							
							\noindent \textbf{Note:} Finer estimates, exploiting decay-properties of $h$, were proven in \cite{PeterDustin}.
							
							\begin{proof}
								Let 
								\begin{equation} A := \Bigl\lbrace Z \in \IR^{6N} :   \Bigl\lvert \frac{1}{N}\, \sum\limits_{i=1}^N h( \phi(z_i)) -  \int h(\phi(z))  f_0(z) \Bigr\rvert \geq N^{-\beta}\Bigr\rbrace. \end{equation} 
								
								\noindent By Markov's inequality, we have for every $M \geq 2$: 
								\begin{equation}\begin{split} \IP_0(A) \leq & \IE_0\Bigl[ N^{2M\beta}\, \Bigl\lvert \frac{1}{N}\, \sum\limits_{1=1}^N h(\phi(z_i)) - \int h(\phi(z)) f_0(z)  \Bigr\rvert^{2M} \Bigr]\\
										= &  \frac{1}{N^{2M(1-\beta)}}\,  \IE \Bigr[  \Bigl(\sum\limits_{i=1}^N \bigl[ h (\phi(z_i)) -  \int h(\phi(z)) f_0(z) \bigr] \Bigr)^{2M}\Bigr]. 
									\end{split}\end{equation}

									\noindent Let $\mathcal{M} := \lbrace \mathbf{k} \in \IN_0^N \mid \lvert \mathbf{k} \rvert = 2M \rbrace$ the set of multiindices $\mathbf{k} = (k_1, k_2, ... , k_N)$ with $\sum\limits_{j=1}^{N} k_j = 2M$. Let
									\begin{equation*} G^\mathbf{k} := \prod \limits_{i=1}^N \bigl[ h (\phi(z_i) -  \int h(\phi(z)) f_0(z) \bigr]^{k_j}. \end{equation*}
									\noindent Then: 
									\begin{equation*} \IE_0\Bigr[ \Bigl(\sum\limits_{i=1}^N \bigl[ h ( \phi(z_i)) - \int h(\phi(z)) f_0(z) \bigr] \Bigr)^{2M}\Bigr] =  \sum\limits_{\mathbf{k} \in \mathcal{M}}\binom{2M}{\mathbf{k}} \, \IE_t(G^\mathbf{k}). \end{equation*}
									
									\noindent Now we observe that  $\IE_0(G^\mathbf{k}) = 0$ whenever there exists a $1 \leq j \leq N$ such that $k_j =1$. This can be seen by integrating the j'th variable first.\\
									
									\noindent For the remaining  terms, we have the bound 
									\begin{align} \int \lvert h(\phi(z)) \rvert^m f_0(z) \, \mathrm{d} z  \lesssim N^{\alpha m} \lVert f_0 \rVert_\infty.\end{align}
									
									\noindent Now, for $\mathbf{k}=(k_1, k_2, ..., k_N) \in \mathcal{M}$, let $\# \mathbf{k}$ denote the number of $k_i$ with $k_i \neq 0$. Note that if $\# \mathbf{k} > M$, we must have $k_i =1$ for at least one $1 \leq i \leq N$, so that $\mathbb{E}_0(G^\mathbf{k}) = 0$. For the other multiindices, we get:
									\begin{equation}\begin{split} \mathbb{E}_0 (G^\mathbf{k}) = \mathbb{E}_0 \Bigl[ &\prod \limits_{i=1}^N \bigl( h ( \phi(q_i)) - \int h(\phi(z)) f_0(z) \bigr)^{k_i}\Bigr]
											\lesssim N^{2M\alpha}.
										\end{split}
									\end{equation}
									\noindent Finally, for any $k\geq 1$, the number of multiindices $\mathbf{k} \in \mathcal{M} $ with $\#\mathbf{k} = j$ is bounded by
									\begin{equation*} \sum\limits_{\#\mathbf{k} = j} 1 \leq \binom{N}{j} (2M)^j  \leq (2M)^{2M} N^j. \end{equation*}
									
									\noindent Thus:
									\begin{align*}\notag \IP_0(A) \lesssim \frac{N^M N^{2M\alpha}}{N^{2M(1-\beta)}} = N^{M(2 (\alpha + \beta) - 1)}\end{align*}	
									and the proposition follows.
								\end{proof}
								
								\noindent We have formulated the proposition with $\phi$ for convenience. The relevant examples for us will be $\phi(z) = z$ and $\phi$ the diffeomorphism defined in \eqref{LCdiffeo}.\\
								
								\noindent In the next section, we will use the law of large numbers to sample the fields on a regular lattice that we introduce on the following definition. 
								
								\begin{Definition}\label{Def:lattice}
									Let $\overline{r}$ as defined in \eqref{xsupport}. For $N \in \IN$ let $\mathcal{G}^N$ be the regular lattice in $[-\overline{r},\overline{r}]^3$ with side length $\frac{d}{N}$. $\mathcal{G}^N$ contains a total of $(3N)^3$ lattice points and for any $x \in [-\overline{r},\overline{r}]^3$, the maximal distance to the next lattice point is at most $\frac{\sqrt{3}}{2} \frac{\overline{r}}{N}.$
								\end{Definition}

								\section{Pointwise estimates}\label{section:pointwiseestimates}
								We will now go deeper into the details of the dynamics to control the difference between mean field and microscopic time-evolution. To this end, we have to control the differences in the electromagnetic fields generated by the (regularized) mean field density $\tilde f^N_t$ and the (smeared) microscopic density $\tilde \mu^N_t[Z] = \mu^N[\Psi_{t,0}(Z)]$ (recall that in view of \eqref{RMV}m the distributions are ``smeared out'' with $\chi^N$ as they enter the field equations.) We will use the decomposition of the fields in terms of Li\'enard-Wiechert distributions introduced in Section \ref{section:lienard}. We will denote by $E_i[\tilde f]$ and $E_i[\tilde \mu]$, $i=0,1,2$ the respective field component generated by $\tilde f^N$, respectively $\tilde\mu^N_t[Z]$.

								\subsection{Controlling the Coulomb term}\label{ControlCoulomb}
								We begin by controlling the contribution of the Coulombic term \eqref{E1term}:
								\begin{align*}
									\bigl \lvert E_1[\tilde f^N](t,x) - E_1[\tilde\mu^N](t,x) \bigr\rvert = \Bigl \lvert \int (\alpha^{-1} Y) *_{t,x} (\mathds{1}_{ t \geq 0} \tilde f^N) \, \dd \xi -  \int (\alpha^{-1} Y) *_{t,x} (\mathds{1}_{t \geq 0} \tilde\mu^N_{(\cdot)}[Z]) \, \dd \xi \Bigr \rvert
								\end{align*}
								with the kernel $\alpha^{-1}$ defined in \eqref{integralkernels}. The expression on the r.h.s. is to be evaluated at $(t,x)$.  Since convolutions commute, we may write
								\begin{align}\notag
									&\bigl \lvert E_1[\tilde f^N](t,x) - E_1[\tilde\mu^N](t,x) \bigr\rvert\\\notag
									= &\Bigl \lvert \chi^N * \Bigl( \int (\alpha^{-1} Y) * (\mathds{1}_{t \geq 0}  f^N) \, \dd \xi -  \int (\alpha^{-1} Y) * (\mathds{1}_{t \geq 0} \mu^N[\Psi_{s,0}(Z)]) \, \dd \xi \Bigr)\Bigr \rvert \\[1.5ex]\label{macroterm}
									\leq  &\Bigl \lvert\chi^N * \Bigl( \int (\alpha^{-1} Y) * (\mathds{1}_{t \geq 0} f^N) \, \dd \xi -  \int (\alpha^{-1} Y) * (\mathds{1}_{t \geq 0} \mu^N[\Phi_{s,0}(Z)]) \, \dd \xi \Bigr)\Bigr \rvert\\\label{microterm}
									+  &\Bigl \lvert\chi^N * \Bigl( \int (\alpha^{-1} Y) * (\mathds{1}_{t \geq 0} \mu^N[\Phi_{s,0}(Z)]) \, \dd \xi -  \int (\alpha^{-1} Y) * (\mathds{1}_{t \geq 0} \mu^N[\Psi_{s,0}(Z)]) \, \dd \xi  \Bigr)\Bigr \rvert
								\end{align}
								where we have inserted the density $\mu^N[\Phi_{s,0}(Z)]$ corresponding to the mean field flow $\Phi_{s,0}(Z) = {}^{N}\Phi_{s,0}(Z)$, in addition to the actual microsocpic density $\mu^N_s[Z] = \mu^N[\Psi_{s,0}(Z)]$.\\
								
								\noindent \textbf{A law of large numbers  bound for \eqref{macroterm}.}
								\noindent Recall from Definition \ref{Def:macroflow}, that $\mu^N[\Phi_{t,0}(Z)] = \varphi^N_{t,0}\# \mu[Z]$, where $\varphi^N_{t,0}$ is the characteristic flow of $f^N_t$. More explicitly, with  $\varphi^N_{t,0}(z_i)=(y^*,\eta^*)(t,z_i)$, we have
								\begin{equation*} 
									\mu^N[\Phi_{t,0}(Z)] = \frac{1}{N} \sum\limits_{i=1}\delta(x - y^*(t, z_i))\delta(\xi - \eta^*(t,z_i)).
								\end{equation*}
								We shall also use the shorthand $y^*_i(t) =y^*(t, z_i), \, \eta^*_i(t) =\eta^*(t, z_i)$. Now we observe that,
								\begin{align*} f^N(t,x,\xi) =  (\varphi^N_{t,0}\# f_0)(x,\xi) &= \int \delta(x - y)\delta(\xi - \eta) (\varphi^N_{t,0}\# f_0)(y,\eta) \dd y \dd \eta\\
									&= \int \delta(x - y^*(t,z))\delta(\xi - \eta^*(t,z)) f_0(z) \dd z.
								\end{align*} 
								\noindent Inserting this into \eqref{macroterm} and performing the $z$-integration last (assuming, for the moment, that the order of integration can be exchanged), we see that 
								\begin{equation*} 
									\IE_0 \Bigl[ \chi^N * \Bigl( \int (\alpha^{-1} Y) * (\mathds{1}_{t \geq 0} f^N) \, \dd \xi -  \int (\alpha^{-1} Y) * (\mathds{1}_{t \geq 0} \mu^N[\Phi_{s,0}(Z)]) \, \dd \xi \Bigl) \Bigr] = 0,
								\end{equation*} 
								where the expectation value is defined with respect to $\otimes^N f_0$. The idea is thus to use the law of large numbers to show that \eqref{macroterm} goes to $0$ in probability.\\
								
								\noindent Recall from \eqref{integralkernels} that:
								\begin{align*} \alpha^{-1}(t,x,\xi) = \frac{(1-v(\xi)^2) (x - t v(\xi))}{(t-v(\xi)x)^2}.
								\end{align*}
								Hence, we compute
								\begin{align*}  &\int (\alpha^{-1} Y) *_{t,x} (\mathds{1}_{t \geq 0} \mu^N[\Phi_{s,0}(Z)]) (t,x)\, \dd \xi \\
									&=  \frac{1}{N} \sum \limits_{i=1}^N \int\limits_{\IR^3\times\IR^3} \int\limits_{0}^t \dd s \dd y \dd \xi \; \delta(y - y^*_i(s)) \delta(\xi - \eta^*_i(s))  \\ 
									&\hspace{3cm} \frac{(1-v(\xi)^2)(x-y-(t-s)v(\eta))}{(t-s-v(\eta)(x-y))^2} \; \frac{\delta(\lvert x - y \rvert - (t-s))}{4\pi \lvert x - y \rvert}\\ 
									&=  \frac{1}{N} \sum \limits_{i=1}^N \int\limits_{0}^t \frac{(1-v(\eta^*_i)^2)(n(x - y^*_i) - v(\eta^*_s))}{4\pi(1-v(\eta^*_i)n(x - y^*_i))^2 \lvert x - y^*_i(s)\rvert^2} \, \delta(\lvert x -y^*_i(s)\rvert - (t-s)) \,  \dd s.
								\end{align*}
								The function $h: s \to \lvert x -y^*_i(s)\rvert - (t-s)$ is differentiable with $h'(s) = 1 - v^*(\eta^*(s))n(x - y^*_i(s))$.  
								If it has a  root in $[0,t]$, we denote it by $t_{ret,i}$, otherwise the integral is zero. Recall that $t_{ret,i} \geq 0 \iff z_i \in B_t(x) \times \IR^3$. Hence, we find: 
								\begin{align}\notag
									\int (\alpha^{-1} Y&) *_{t,x} (\mathds{1}_{t \geq 0} \mu^N[\Phi_{s,0}(Z)])(t,x) \, \dd \xi \\
									= & \frac{1}{N} \sum \limits_{i=1}^N \frac{(1-v(\eta^*_i)^2)(n(x - y^*_i) - v(\eta^*_i))}{4\pi(1-v(\eta^*_i)n(x - y^*_i))^3 \lvert x - y^*_i(s)\rvert^2}\, \mathds{1}_{\lbrace s \geq 0 \rbrace}\;  \Biggr \rvert_{s=t_{ret,i}}\\
									= &  \frac{1}{N} \sum\limits_{i=1}^N \mathds{1}_{\lbrace{z_i \in B_t(x) \times \IR^3 }\rbrace}\,  k\bigl(x - y^*(t_{ret,i}, z_i), \eta^*(t_{ret,i}, z_i) \bigr),
								\end{align}
								where we have introduced the kernel 
								\begin{equation}\label{kernelk} k(x,\xi):= \frac{(1-v(\xi)^2)(n(x) - v(\xi))}{4\pi(1-v(\xi)\cdot n(x))^3 \lvert x \lvert^2}.\end{equation}
								Furthermore, according to Lemma \ref{Lemma:lightconedistribution}, 
								\begin{align*}
									&\int\limits (\alpha Y) *_{t,x} (\mathds{1}_{t\geq 0} f^N) (t,x, \eta)\,\dd \eta\, \\
									&= \int\limits_{B_t(x) \times \IR^3}\frac{\alpha^{-1}(t-s, x -y^*(s,z), \eta^*(s,z))}{\lvert x - y^*(s,z) \rvert (1 -  n(x-y^*(s,z)) \cdot v(\eta^*(s,z)) } \Biggl \lvert_{s = tret(z)} \; f_0(z) \, \dd z \\
									& = \int\limits_{B_t(x) \times \IR^3}  \frac{(1-v(\eta^*(s,z))^2)(n(x - y^*(s,z)) - v(\eta^*(s,z)))}{4\pi(1-v(\eta^*(s,z))n(x - y^*(s,z)))^3 \lvert x - y^*(s,z)\rvert^2} \Biggr \rvert_{s=t_{ret(z)}}  \, f_0(z) \, \dd z\\
									&= \int\limits_{B_t(x) \times \IR^3} k(x-y^*(t_{ret(z)}, z), \eta^*(t_{ret}(z), z)) \, f_0(z) \, \dd z.
								\end{align*}
								\noindent (In fact, we could have also applied the same identity \eqref{LCdistribution2} to $\mu^N[\Phi_{t,0}(Z)]$).\\
								\noindent Now note that on the support of $f$, we have
								\begin{equation}\label{kbounds}
									\lvert k(x,\xi) \rvert \leq \frac{1}{2\pi(1-\overline{v})^3 \lvert x \lvert^2},
								\end{equation}
								and thus, according to Lemma \ref{chibounds},
								\begin{equation} 
									\lvert \tilde k(x, \xi) \rvert = \lvert \chi^N*_x k (x, \xi) \rvert \lesssim r_N^{-2}, \;\; \forall x \in \IR^3, \lvert \xi \rvert \leq  \overline{\xi}
								\end{equation}
								where we have applied the mollifier $\chi^N$. In total, we have found that \eqref{macroterm} is of the form
								\begin{equation*} \biggl\lvert \frac{1}{N} \sum\limits_{i=1}^N h(\phi(z_i)) - \int h(\phi(z)) \, \dd f_0(z) \biggr \rvert \end{equation*}
								with $h(y,\eta) = \tilde k (x-y, \eta)$ and $\phi$ the diffeomorphism defined in Lemma \ref{Lemma:lightconedistribution} and $f_0$ restricted to $\mathrm{B}(t ;x) \times \IR^3$. Hence, we can use the law of large numbers in the form of Proposition \ref{Prop:LLN}  to conclude the following:\\
								\begin{Lemma}
									\noindent Let $A^1_t$ be the ($N$ and $t$ dependent) set defined by
									\begin{equation}\label{A1} A^1_t := \lbrace Z \in \IR^3\times\IR^3 \mid \eqref{macroterm} < N^{-1/3} \text{ for all } x \in \mathcal{G}^N \rbrace. \end{equation}
									Then there exists $C_1 > 0$ such that $\IP_0(A^1_t) \geq 1- \frac{C_1}{N^1}$.
								\end{Lemma}
								\begin{proof}
									Let $\mathcal{G}^N$ the lattice defined in \ref{Def:lattice} and $x_k \in \mathcal{G}^N$.  We want to apply Proposition \ref{Prop:LLN} with $h(y, \eta) = \tilde k (x_k-y, \eta)$ and $\phi$ as in $\eqref{LCdiffeo}$. Since $\lvert h \rvert \lesssim r_N^{-2} \leq N^{2\gamma}$, with $\gamma < \frac{1}{12}$, we can choose $\beta = \frac{1}{3}$. Thus, by Prop. \ref{Prop:LLN}, there exists a constant $C> 0$ such that
									\begin{equation*} \begin{split}\IP_0 \Bigl[\Bigl \lvert\chi^N*\Bigl( \int (\alpha^{-1} Y) * \mathds{1}_{t \geq 0}( f^N -  \mu^N[\Phi_{s,0}(Z)]\bigl) \, \dd \xi \Bigr) (t, x_k) \Bigr \rvert \geq N^{-\frac{1}{3}} \Bigr]
											\leq \frac{C}{N^4}.\end{split}\end{equation*}
									Since the lattice $\mathcal{G}^N$ contains $(3N)^3$ points, we have
									\begin{align*} &\IP_0 \bigl [\exists x_k \in C^N : \eqref{macroterm} \geq N^{-\frac{1}{3}}\bigr] \\ 
										&\leq \sum\limits_{x_k \in \mathcal{G}^N}\IP_0 \Bigl[\Bigl \lvert\chi^N*\Bigl( \int (\alpha^{-1} Y) * (\mathds{1}_{t \geq 0} f^N - \mathds{1}_{t \geq 0} \mu^N[\Phi_{t,0}(Z)]\bigl) \, \dd \xi \Bigr) (t, x_k) \Bigr \rvert \geq N^{-\frac{1}{3}} \Bigr]\\
										&\leq (3N)^3 \frac{C}{N^4} \leq \frac{27C}{N}. \end{align*}
								\end{proof}
								
								\noindent \textbf{A Lipschitz bound bound for \eqref{microterm}.} We now have to control \eqref{microterm}, i.e. the difference of the field components $E_1$ generated by the mean field trajectories $(y^*_i, \eta^*_i)_{i=1,..,N}$ on the one hand and the true microscopic trajectories $(x^*_i, \xi^*_i)_{i=1,..,N}$ on the other hand. To this end, we want to establish a local Lipschitz bound for the kernel \eqref{kernelk}.
								\begin{Lemma}[Local Lipschitz bound]\label{Lemma:gbound}
									There exists constants  $b_1, b_2 >0$ and functions
									\begin{equation}\label{Def:g}
										g_1(x) := \frac{b_1}{(1- \overline v)^3}  \begin{cases}  r_N^{-3} &; \lvert x \rvert < \frac{2r_N}{1 - \overline{v}}\\  \lvert x \rvert^{-3} &; \lvert x \rvert \geq \frac{2r_N}{1 - \overline{v}}
										\end{cases},\;\;\; g_2(x):= \frac{b_2}{(1- \overline v)^4} \begin{cases}  r_N^{-2} &; \lvert x \rvert < r_N\\  \lvert x \rvert^{-2} &; \lvert x \rvert \geq r_N\end{cases}.
									\end{equation}         
									
									\noindent such that for all $z_1= (x_1, \xi_1), z_2= (x_2, \xi_2)$ with $\lvert \xi_1\rvert, \lvert \xi_2 \rvert \leq \overline{\xi}$ and  $\lvert x_1 - x_2 \rvert < \frac{r_N}{1- \overline{v}}, \, \overline{v} = \lvert v(\overline{\xi})\rvert$:
									\begin{equation}\label{gbound} \lvert \tilde k(x_1,\xi_1) - \tilde k (x_2, \xi_2)\rvert_\infty \leq g_1(x_1) \, \lvert x_1 - x_2 \rvert_\infty + g_2(x_1) \, \lvert \xi_1 - \xi_2 \rvert_\infty. \end{equation}
									
								\end{Lemma}
								
								\begin{proof}We have 
									\begin{align*} 	
										\lvert \tilde k(x_1,\xi_1) - \tilde k (x_2, \xi_2)\rvert_\infty\leq  \lvert \tilde k(x_1,\xi_2) - \tilde k (x_2, \xi_2)\rvert_\infty + \lvert \tilde k(x_1,\xi_1) - \tilde k (x_1, \xi_2)\rvert_\infty,
									\end{align*}
									hence, there exists $y$ between $x_1$ and $x_2$ and $\zeta$ between $\xi_1$ and $\xi_2$ such that
									\begin{align*} 	
										\lvert \tilde k(x_1,\xi_1) - \tilde k (x_2, \xi_2)\rvert_\infty \leq  \lvert \nabla_x \tilde k(y, \xi_2) \rvert_\infty \lvert x_1 - x_2 \rvert_\infty   +  \lvert \nabla_\xi \tilde k(x_1, \zeta) \rvert_\infty \lvert \xi_1 - \xi_2 \rvert_\infty.
									\end{align*}
									Now one checks that
									\begin{equation*} \lvert \nabla_\xi k (x,\xi)\rvert_\infty \leq \frac{18}{(1-\overline{v})^4 \lvert x \rvert^2}, \end{equation*}
									so that according to Lemma  \ref{chibounds}, there exists $b_2 >0$ such that
									\begin{equation}\label{Dk1}\lvert \nabla_\xi \tilde k (x,\xi)\rvert_\infty \leq  \frac{b_1}{(1-\overline{v})^4} \min \lbrace r_N^{-2},  {\lvert x \rvert^{-2}}  \rbrace. \end{equation}
									For the difference in the $x$-coordinates, we get from \eqref{kbounds} and Lemma \ref{chibounds} a constant $b >0$ such that
									\begin{equation}\label{Dk2}\lvert \nabla_x \tilde k (x,\xi)\rvert_\infty \leq \frac{b}{(1 - \overline{v})^3} \min \lbrace r_N^{-3}, {\lvert x \rvert^{-3}} \rbrace. \end{equation}
									
									\noindent Thus, for $\lvert x_1\rvert < \frac{2r_N}{1- \overline{v}}$, a bound of the form \eqref{gbound} certainly holds, since the derivative is bounded by $\frac{b}{(1 - \overline{v})^3}\, r_N^{-3}$. For $\lvert x_1 \rvert > \frac{2r_N}{1- \overline{v}}$ and $\lvert x_1 - x_2 \rvert <\frac{r_N}{1- \overline{v}}$ we observe that $\lvert s x_1 + s(x_2 -x_1)\rvert \geq \frac{\lvert x_1 \rvert}{2}, \forall s \in [0,1]$, so that
									$\frac{1}{\lvert s x_1 + s(x_2 -x_1)\rvert^3} \leq \frac{8}{\lvert x_1 \rvert^3}$. Setting $b_1 := 8b$, the statement follows. 
								\end{proof}
								
								\noindent Now recall that as long as $J^N_t(Z) <1$, the trajectories are close as per \eqref{Def:J2}. More precisely, $J^N_t(Z) <1 \Rightarrow \sup\limits_{0\leq s\leq t} \lvert {}^N\Phi_{t,0}(Z) - {}^N\Psi_{t,0}(Z)\rvert_\infty  < N^{-\delta} \leq N^{-\gamma} \leq r_N$. This implies, in particular, $\lvert x^*(s,z_i) - y^*(s,z_i) \rvert < r_N$ as well as $\lvert \xi^*(s,z_i) \rvert < \overline{\xi}$ for $0 \leq s  \leq t $ and all $1 \leq i \leq N$. Moreover, with Lemma \ref{Lemma:closetret} we have for any fixed $(t,x) \in \IR^+\times \IR^3$:
								\begin{equation} \lvert x^*_i(t^x_{ret,i}) -  y^*_i(t^y_{ret,i}) \rvert \leq \frac{r_N}{1 - \overline{v}}, \end{equation}
								where $t^x_{ret,i}$ and $t^y_{ret,i}$ denote the retarded time of the trajectory $x^*_i(s)$, respectively $y_i^*(s)$, with respect to the space-time point $(t,x)$. Hence, we can apply the previous Lemma and find that \eqref{microterm} is bounded by
								\begin{align}\notag
									\frac{1}{N} \sum\limits_{i=1}^N \mathds{1}_{\lbrace t_{ret} \geq 0 \rbrace}\,\Bigl\lvert  &\tilde k\bigl(x - x^*(t^x_{ret,i}, z_i), \xi^*(t^x_{ret,i}, z_i)) - \tilde k\bigl(x - y^*(t^y_{ret,i}, z_i), \eta^*(t^x_{ret,i}, z_i))\Bigr\rvert\\\notag
									\leq  \frac{1}{N} \sum\limits_{i=1}^N \mathds{1}_{\lbrace t_{ret} \geq 0 \rbrace} \Bigl( &g_1(x - y_i^*(t^y_{ret,i})) \, \lvert  x_i^*(t^x_{ret,i}) -  y_i^*(t^y_{ret,i})\rvert_\infty\\\notag
									+ &g_2(x - y_i^*(t^y_{ret,i})) \, \lvert  \xi_i^*(t^x_{ret,i}) -  \eta_i^*(t^y_{ret,i})\rvert_\infty \Bigr) \\ \label{gterm1}
									\leq \Bigl (\frac{1}{N} \sum\limits_{i=1}^N \mathds{1}_{\lbrace t_{ret} \geq 0 \rbrace}& g_1(x - y_i^*(t^y_{ret,i})) \Bigr) \, \frac{1}{1 - \overline{v}} \, \sup\limits_{0\leq s\leq t} \lvert {}^N\Phi^1_{s,0}(Z) - {}^N\Psi^1_{s,0}(Z) \rvert_\infty \;\;\\\notag
									+ \Bigl (\frac{1}{N} \sum\limits_{i=1}^N \mathds{1}_{\lbrace t_{ret} \geq 0 \rbrace} &g_2(x - y_i^*(t^y_{ret,i})) \Bigr) \cdot\\\label{g2term1}
									 \sup\limits_{0\leq s\leq t} \Bigl( \lvert {}^N\Phi^2_{s,0}(Z) - &{}^N\Psi^2_{s,0}(Z) \rvert_\infty + \frac{L}{1 - \overline{v}} \, \lvert {}^N\Phi^1_{s,0}(Z) - {}^N\Psi^1_{s,0}(Z) \rvert_\infty \Bigr).
								\end{align}
								\noindent For the last inequality, we used Lemma \ref{Lemma:closetret} and the bound \eqref{LipschitzL} on the mean field force to account for the fact that the distance $\lvert  x_i^*(t^x_{ret,i}) -  y_i^*(t^y_{ret,i})\rvert$, respectively $\lvert  \xi_i^*(t^x_{ret,i}) -  \eta_i^*(t^y_{ret,i})\rvert$, involves to different retarded times. Now, we want to estimate $\frac{1}{N} \sum\limits_{i=1}^N \mathds{1}_{\lbrace t_{ret} \geq 0 \rbrace}g_j(x - y_i^*(t^y_{ret,i})),  j=1,2$ by its expectation value w.r.to $f_0$. In view of Lemma \ref{Lemma:lightconedistribution}, we write:
								\begin{align*} \notag & \frac{1}{N} \sum\limits_{i=1}^N \mathds{1}_{\lbrace t_{ret} \geq 0 \rbrace}g_j(x - y_i^*(t^y_{ret,i}))\\ 
									\leq& \Bigl\lvert \frac{1}{N} \sum\limits_{i=1}^N \mathds{1}_{\lbrace t_{ret} \geq 0 \rbrace}g_j(x - y_i^*(t^y_{ret,i})) - \hspace{-5mm} \int\limits_{\mathrm{B}_t(x) \times \IR^3} \hspace{-5mm}g_j(x - y)(1- n(x-y)v(\eta))  f^N(t - \lvert x - y\rvert ,y,\eta) \Bigr\rvert\\ \notag 
									+&  \Bigl\lvert \int\limits_{\mathrm{B}_t(x) \times \IR^3} g_j(x - y) (1- n(x-y) v(\eta)) f^N(t - \lvert x -y\rvert, y,\eta) \dd y \dd \eta \Bigr\rvert.
								\end{align*}
								For the last term, we recall the bounds from \eqref{Def:g} and estimate, using $\lvert 1-n\cdot v \rvert\leq 2$,
								\begin{align}\notag
									\Bigl\lvert \int\limits_{\mathrm{B}_t(x) \times \IR^3}& g_1(x - y) (1- n(x-y) v(\eta)) f^N(t  - \lvert x -y\rvert, y,\eta) \dd y \dd \eta \Bigr\rvert\\\notag
									&\lesssim \int\limits_{\lvert x-y \rvert \leq t} g_1(x - y) \rho[f^N](t -\lvert x- y \rvert, y) \dd y \\\notag
									&\leq \sup\limits_{0 \leq s \leq t} \lVert \rho[f^N](s, \cdot)\rVert_\infty\, \Bigl(\int\limits_{\lvert y \rvert \leq  \frac{2 r_N}{1-\overline{v}}}  g_1(y) \, \dd^3 y + \int\limits_{\frac{2 r_N}{1-\overline{v}} < \lvert y \rvert \leq t} g_1(y) \, \dd^3 y \Bigr)\\\notag
									&\lesssim   \;\;\,  C_0 \Bigl( \int\limits_{\lvert y \rvert \leq \frac{2 r_N}{1-\overline{v}}} r_N^{-3}\, \dd^3 y + \int\limits_{  \frac{2 r_N}{1-\overline{v}} < \lvert y \rvert \leq t}  \lvert y \rvert^{-3} \, \dd^3 y\Bigr)\\[1.2ex]\label{gterm2}
									& \lesssim  \;\;\,   C_0  \, (1 + \log(r_N^{-1}) + \log(T)),
								\end{align}
								and for $g_2$:
								\begin{align}\notag
									\Bigl\lvert \int\limits_{\mathrm{B}_t(x) \times \IR^3}& g_2(x - y) (1- n(x-y) v(\eta)) f^N(t  - \lvert x -y\rvert, y,\eta) \dd y \dd \eta \Bigr\rvert\\\notag
									&\lesssim \int\limits_{\lvert x-y \rvert \leq t} g_2(x - y) \rho[f^N](t -\lvert x- y \rvert, y) \dd y \\\notag
									&\lesssim  \sup\limits_{0 \leq s \leq t} \lVert \rho[f^N](s, \cdot)\rVert_\infty\, \int\limits_{\lvert y \rvert \leq  t} \lvert y \rvert^{-2}  \dd^3 y \\\label{g2term2}
									&\lesssim   \;\;\,  C_0 T.
								\end{align}
								
								\noindent It remains to show that the difference
								\begin{equation}\label{deribound}
									\biggl\lvert \frac{1}{N} \sum\limits_{i=1}^N \mathds{1}_{\lbrace t_{ret} \geq 0 \rbrace} g_j(x - y_i^*(t_{ret,i})) - \int g_j(x - y)(1- nv)  f^N(t -\lvert x -y \rvert,y,\xi)  \biggr\rvert
								\end{equation}
								is typically small. According to part 1) of Lemma \ref{Lemma:lightconedistribution}, \eqref{deribound} can be written as
								\begin{equation*} 
									\biggl\lvert \frac{1}{N} \sum\limits_{i=1}^N \mathds{1}_{\lbrace{z_i \in B_t(x) \times \IR^3 }\rbrace} g_j(x- \pi_x \phi(z_i)) - \int \mathds{1}_{\lbrace{z \in B_t(x) \times \IR^3 }\rbrace} g_j(x - z) \phi \# f_0(z) \dd z \biggr\rvert,
								\end{equation*}
								where $\pi_x(x,\xi) = x$ is the projection on the spatial coordinates and we used the fact that $t_{ret}(z)\geq 0 \iff z \in B(t,x) \times \IR^3$. Hence, we can apply again the law of large numbers.\\
								
								\noindent For any $x \in \mathcal{G}^N$, we consider $h:\IR^6\to \IR, z \mapsto \mathds{1}_{\lbrace{\phi^{-1}(z) \in B_t(x) \times \IR^3 }\rbrace} g_j(x - \pi_x z)$. This function is bounded as $\lvert h \rvert \lesssim r_N^{-3} \leq N^{3\gamma}$ with $\gamma < \frac{1}{12}$. Applying Proposition \ref{Prop:LLN} with $\phi$ as in \eqref{LCdiffeo}, $\alpha = 3\gamma$ and $\beta =0$, we find
								\begin{equation*}\IP_0\Bigr[\Bigl\lvert \frac{1}{N} \sum\limits_{i=1}^N \mathds{1}_{\lbrace t_{ret} \geq 0 \rbrace} g_j(x - y_i^*(t_{ret,i})) - \int g_j(x - y)(1- nv)  f^N(t -\lvert x -y \rvert,y,\xi)  \Bigr\rvert > 1\Bigr] \lesssim N^{-4} \end{equation*} 
								and thus $\IP_0\bigl[\exists x_k \in \mathcal{G}^N \mid \eqref{deribound} > 1 ] \lesssim N^{-1}$, for $j=1,2$, since the grid $\mathcal{G}^N$ consists of $(3N)^3$ points. We define the ($N$ and $t$ dependent) set 
								\begin{equation}\label{gterm3} A^2_t := \lbrace Z \in \IR^3\times\IR^3 \mid \eqref{deribound} \leq 1, j=1,2 \; \forall x \in \mathcal{G}^N \rbrace. \end{equation}
								Then  there exists $C_2 > 0$ such that $\IP(A^2_t) \geq 1- \frac{C_2}{N}$.\\
								
								\noindent For the magnetic field component $B_1$, the proof works analogously, since the corresponding kernel $n \times \alpha^{-1}$ has the same bounds and regularity properties. 
								
								\subsection{Controlling the radiation term}
								We now consider the contribution of the radiation term $E_2$. The corresponding kernel is less singular in the near-field, but depends on the acceleration of the particles. From \eqref{E2term}:
								\begin{align}\notag
									&\lvert E_2[\tilde f^N](t,x) - E_2[\tilde\mu^N](t,x)\rvert\\\notag
									&= \Bigl \lvert \int (\nabla_\xi \alpha Y) *(\tilde K[\tilde f^N] \mathds{1}_{t \geq 0} \tilde f^N) \, \dd \xi -  \int (\nabla_\xi \alpha Y) * (\tilde K[\tilde \mu^N] \mathds{1}_{t \geq 0} \tilde\mu^N) \, \dd \xi \Bigr \rvert \\\label{macroterm2}
									&\leq  \Bigl \lvert \int (\nabla_\xi \alpha Y) * (\tilde K[\tilde f^N]  \mathds{1}_{t \geq 0} \tilde f^N) \, \dd \xi -  \int (\nabla_\xi \alpha Y) * (\tilde K[\tilde f^N]  \mathds{1}_{t \geq 0} \tilde\mu^N[\Psi_{s,0}(x)]) \, \dd \xi \Bigr \rvert \\\label{microterm2}
									& +  \Bigl \lvert \int (\nabla_\xi \alpha Y) * (\tilde K[\tilde f^N]  - \tilde K[\tilde \mu^N] ) ( \mathds{1}_{t \geq 0} \, \tilde\mu^N[\Psi_{s,0}(x)]) \; \dd \xi \Bigr\rvert,
								\end{align}
								where we use the regularized distributions and the corresponding regularized forces $K[\tilde f^N]$, respectively $K[\tilde \mu^N]$ in view of \eqref{RMV}. The integrals on the r.h.s. are to be evaluated at $(t,x)$. For the second term \eqref{microterm2}: 
								\begin{align*} &\Bigl \lvert \int (\nabla_\xi \alpha \, Y) * (\tilde K[\tilde f^N]  - \tilde K[\tilde \mu^N]) ( \mathds{1}_{t \geq 0} \, \tilde\mu^N) \; \dd \xi \Bigr\rvert\\
									&= \Bigl\lvert \frac{1}{N}\sum\limits_{i=1}^N \int\limits_{0}^t \int\limits_{S^2} (t-s) \nabla_\xi \alpha (t-s,  \omega(t-s), \xi^*_i(s))\\
									& \hspace{2.5cm} (\tilde K[\tilde f^N]  - \tilde K[\tilde \mu^N]) (s, x - \omega(t-s), \xi^*_i(s)) \chi^N(x - \omega(t-s) - x^*_i(s)) \, \dd \omega \dd s\Bigr\rvert\\
									&\leq   \frac{1}{N}\sum\limits_{i=1}^N \int\limits_{0}^t \int\limits_{S^2}\Bigl\lvert (t-s) \nabla_\xi \alpha (t-s,  \omega(t-s), \xi^*_i(s)) \Bigr\rvert\\
									& \hspace{2.5cm}\Bigl\lvert (\tilde K[\tilde f^N]  - \tilde K[\tilde \mu^N]) (s, x - \omega(t-s), \xi^*_i(s)) \Bigr\rvert \chi^N(x - \omega(t-s) - x^*_i(s)) \, \dd \omega \dd s.
								\end{align*}
								
								\noindent Now, recall from $\eqref{nablaalpha}$:
								\begin{equation*}
									(\nabla_\xi \alpha^0)^i_j (t, x, \xi) = \frac{t (t-v\cdot x)(v_jv^i - \delta^i_j) + (x_j - t v_j)(x^i - (v \cdot x) v^i)}{\sqrt{1+ \lvert \xi \rvert^2} (t - v \cdot x)^2}
								\end{equation*}
								and thus 
								\begin{align*}
									(t-s) \nabla_\xi \alpha (t-s,  \omega(t-s), \xi^*) = \frac{(1-v\cdot \omega)(v_jv^i - \delta^i_j) + (\omega_j -  v_j)(\omega^i - (v \cdot \omega) v^i)}{\sqrt{1+ \lvert \xi \rvert^2} (1 - v \cdot \omega)^2}.
								\end{align*}
								Since the vectors appearing in the nominator are all of norm $1$ or smaller, we can estimate
								\begin{equation}
									\lvert (t-s) \nabla_\xi \alpha (t-s,  \omega(t-s), \xi^*) \rvert \leq \frac{8}{(1-\overline{v})^2}.
								\end{equation} 
								Moreover, we observe that $\frac{1}{N}\sum\limits_{i=1}^N \chi^N(x - \omega(t-s) - x^*_i(s))$
								is nothing else than the (smeared) microscopic charge density  $\tilde \rho[\mu^N[Z]](s, x - \omega(t-s))$. In total, we can thus write
								\begin{align}
									\notag
									&\Bigl \lvert \int (\nabla_\xi \alpha \, Y) * (\tilde K[\tilde f^N]  - \tilde K[\tilde \mu^N]) ( \mathds{1}_{t \geq 0} \, \tilde\mu) \; \dd \xi \Bigr\rvert\\\notag
									& \leq \frac{8}{(1-\overline{v})^2} \int\limits_{0}^t  \int\limits_{S^2} \bigl\lvert E[\tilde f^N](s,x - \omega(t-s)) - E[\tilde\mu^N](s, x - \omega(t-s)) \bigr\rvert \\\notag 
									&\hspace{1cm} + \bigl\lvert B[\tilde f^N](s, x - \omega(t-s)) - B[\tilde\mu^N](s, x - \omega(t-s)) \bigr\rvert \,\tilde\rho[\mu] (s, x - \omega(t-s)) \dd \omega \dd s\\\label{KGronwall}
									&\lesssim \frac{\lVert \tilde\rho[\mu]\rVert_{L^\infty([0,T]\times \IR^3)}}{(1-\overline{v})^2} \int\limits_{0}^t \lVert E[\tilde f^N](s) - E[\tilde\mu^N](s) \rVert_{L^\infty(\mathrm B(\overline{r}))} + \lVert B[\tilde f^N](s) - B[\tilde\mu^N](s) \rVert_{L^\infty(\mathrm B(\overline{r}))}  \dd s,
								\end{align}
								where in the last line, we used the fact that $\supp \tilde\rho[\mu] (s) \subseteq \mathrm{B}(\overline{r};0), \; \forall s \leq T$.\\

								\noindent For \eqref{macroterm2} we write
								\begin{align}\notag
									& \Bigl \lvert \int (\nabla_\xi \alpha \,  Y) * (\tilde K[\tilde f^N] \mathds{1}_{t \geq 0} \tilde f^N) \, \dd \xi -  \int(\nabla_\xi \alpha \,  Y) * (\tilde K[\tilde f^N] \mathds{1}_{t \geq 0} \tilde\mu^N[\Psi_{t,0}(Z)]) \, \dd \xi \Bigr \rvert \\\label{macroterm3}
									\leq&  \Bigl \lvert \int (\nabla_\xi \alpha \,  Y) * (\tilde K[\tilde f^N] \mathds{1}_{t \geq 0} \tilde f) \, \dd \xi -  \int(\nabla_\xi \alpha \,  Y) * (\tilde K[\tilde f^N] \mathds{1}_{t \geq 0} \tilde\mu^N[\Phi_{t,0}(Z)]) \, \dd \xi \Bigr \rvert\\\label{microterm3}
									+&  \Bigl \lvert \int(\nabla_\xi \alpha \,  Y) * (\tilde K[\tilde f^N] \mathds{1}_{t \geq 0} \tilde\mu^N[\Phi_{t,0}(Z)]) \, \dd \xi -  \int (\nabla_\xi \alpha \,  Y) * (\tilde K[\tilde f^N] \mathds{1}_{t \geq 0} \tilde\mu^N[\Psi_{t,0}(Z)]) \, \dd \xi \Bigr \rvert.
								\end{align}
								\noindent We evaluate 
								\begin{align*} 
									\int(\nabla_\xi \alpha \,  Y) * (\tilde K[\tilde f^N] \mathds{1}_{t \geq 0} \tilde\mu^N[\Phi_{t,0}(Z)])
									= \frac{1}{N}\sum\limits_{i=1}^N \mathds{1}_{\lbrace t_{ret,i}>0 \rbrace}\kappa(t_{ret,i}, y^*(t_{ret,i}), \eta^*(t_{ret,i}) ) \end{align*}
								with kernel
								\begin{equation}\begin{split} \kappa(s,y,\eta) &:= \frac{(\tilde K[\tilde f](s,y,\eta) \cdot v(\eta)) v(\eta) - \tilde K[\tilde f](s,y,\eta))}{\sqrt{1 + \eta^2}(1 - v(\eta) \cdot n(x-y))^2 \lvert x - y\rvert}\\[1.5ex]
										&+ \frac{\tilde K[\tilde f](s,y,\eta) \cdot (n(x-y)-v(\eta)) \bigl(n(x-y) - (v \cdot n) v(\eta) \bigr)}{\sqrt{1 + \eta^2}(1 - v(\eta) \cdot n(x-y))^2 \lvert x - y\rvert}.
									\end{split}\end{equation}
									With $L$ as in \eqref{LipschitzL}, the function $\kappa$ satisfies
									\begin{align}
										\lvert \kappa(s,y,\eta) \rvert &\lesssim \frac{\lvert  \tilde K[\tilde f^N](s,y,\eta) \rvert}{(1-\overline{v})^2 \lvert x - y \rvert} \leq \frac{L}{(1-\overline{v})^2 \lvert x - y \rvert}\\[1.5ex]\notag \;\;\; \vert \nabla_{x,\xi} \kappa (s,y,\eta) \rvert &\lesssim \frac{\lvert  \nabla_{x,\xi} \tilde K[\tilde f^N](s,y,\eta) \rvert}{(1-\overline{v})^3 \lvert x - y \rvert} + \frac{\lvert  \tilde K[\tilde f^N](s,y,\eta) \rvert}{(1-\overline{v})^2 \lvert x - y \rvert^2}\\
										& \leq \frac{L}{(1-\overline{v})^3} \Bigl(  \frac{1}{\lvert x - y \rvert} + \frac{1}{\lvert x - y \rvert^2} \Bigr).
									\end{align}
									
									\noindent Now we proceed along the lines of section \ref{ControlCoulomb}, simplified by the fact that the kernel is homogeneous of degree $-1$ (rather than $-2$) in $x$.\\
									
									\noindent Let $A^3_t$ be the ($N$ and $t$ dependent) set defined by
									\begin{equation}\label{A3} A^3_t := \lbrace Z \in \IR^3\times\IR^3 \mid \eqref{macroterm3} \leq N^{-1/4}  \text{ for all } x \in \mathcal{G}^N  \rbrace.\end{equation}
									Then there exists $C_3 > 0$ such that $\IP(A^3_t) \geq 1- \frac{C_3}{N}$.\\
									
									\noindent For \eqref{microterm3}, we introduce a function $g_3 \lesssim \min\lbrace r_N^{-2} , \lvert x \rvert^{-1} + \lvert x \rvert^{-2} \rbrace$ such that
									\begin{equation} \lvert \tilde \kappa (t, x_1,\xi_1) - \tilde \kappa (t, x_2, \xi_2)\rvert_\infty \leq g_3(x_1) \, \lvert (x_1,\xi_1) - (x_2, \xi_2) \rvert_\infty, \end{equation}
									for all $t \leq T$, $\lvert \xi_1\rvert, \lvert \xi_2 \rvert \leq \overline{\xi}$ and  $\lvert x_1 - x_2 \rvert < \frac{r_N}{1- \overline{v}}$ (c.f. Lemma \ref{Lemma:gbound}). With this, we find that
									\begin{align}\label{g3term1} \eqref{microterm3} \leq \Bigl (\frac{1}{N} \sum\limits_{i=1}^N \mathds{1}_{\lbrace t_{ret} \geq 0 \rbrace} g_3(x - y_i^*(t^y_{ret,i})) \Bigr) \, \frac{L}{1 - \overline{v}} \, \sup\limits_{0\leq s\leq t} \lvert {}^N\Phi_{s,0}(Z) - {}^N\Psi_{s,0}(Z) \rvert_\infty.
									\end{align}
									
									\noindent In contrast to \ref{ControlCoulomb}, we do not have to treat distances in physical space and momentum space separately, other than that, the argument is the same. We estimate the $g_3$ term by
									\begin{align}\notag 
										&\Bigl\lvert \frac{1}{N} \sum\limits_{i=1}^N \mathds{1}_{\lbrace t_{ret} \geq 0 \rbrace}\, g_3(x - y_i^*(t^y_{ret,i})) \Bigr\rvert \\\label{g1term}
										& \leq  \biggl\lvert \frac{1}{N} \sum\limits_{i=1}^N \mathds{1}_{\lbrace{z_i \in B_t(x) \times \IR^3 }\rbrace}\, g_3(x- \pi_x \phi(z_i)) - \int \mathds{1}_{\lbrace{z \in B_t(x) \times \IR^3 }\rbrace} g_3(x - z) \phi \# f_0(z) \dd z \biggr\rvert \\\label{g1term2}
										&+ \biggl\lvert \int \mathds{1}_{\lbrace{z \in B_t(x) \times \IR^3 }\rbrace}\, g_3(x - z) \phi \# f_0(z) \dd z \biggr\rvert.
									\end{align}
									Since $g_3 \lesssim \min\lbrace r_N^{-2} , \lvert x \rvert^{-1} + \lvert x \rvert^{-2} \rbrace$, one checks that $\eqref{g1term2} \lesssim C_0 (1 + T^2)$. Now we define the ($N$ and $t$ dependent) set
									\begin{equation} 
										\label{A4} A^4_t := \lbrace Z \in \IR^3\times\IR^3 \mid \eqref{g1term} \leq 1  \text{ for all } x \in \mathcal{G}^N  \rbrace.\end{equation}
									According to Proposition \ref{Prop:LLN}, there exists a constant $C_4 >0$ such that $\IP_0(A) \geq 1 - \frac{C_4}{N}$. For  $Z \in A^4_t, J^N_t(Z)<1 $, we thus have $\eqref{microterm3} \lesssim  \sup\limits_{0\leq s\leq t} \lvert {}^N\Phi^1_{s,0}(Z) - {}^N\Psi^1_{s,0}(Z) \rvert_\infty$.\\
									
									\noindent For the magnetic field component $B_2$, the proof works analogously, since the corresponding kernel $\nabla_\xi n \times \alpha^{0}$ has the same bounds and regularity properties.

									\subsection{Controlling shock waves}
									We now consider the term \eqref{E0'term}. We compute 
									\begin{align*} E'_0(t,x) &= \; \int (\alpha^0\, Y)(t, \cdot,\xi) *_x\chi^N*_x f_0 (x, \xi) \dd \xi\\
										&= \frac{t}{4\pi} \int \frac{\omega - v}{1 - v\cdot \omega} \, \chi^N(x-y-wt)\, f_0(y, \xi)\, \dd w \dd y \dd \xi\\
										& = \; \int h(t, x-y) \,f_0(y,\xi) \dd y \dd \xi,
									\end{align*}
									with
									\begin{align}\label{hshocks}
										h(t, x, \xi)= \frac{t}{4\pi} \int\limits_{S^2} \frac{\omega - v}{1 - v\cdot \omega} \, \chi^N(x-wt).
									\end{align}
									\noindent This function satisfies
									\begin{equation}\label{swkernelbound} \lvert h (t, x, \xi) \rvert \lesssim \frac{t}{1 - \overline{v}}\, r_N^{-3}.  \end{equation} 
									\noindent We have to control the difference
									\begin{align}\notag 
										& \bigl\lvert E'_0[\tilde \mu^N_0[Z]](t,x) - E'_0[\tilde f_0](t,x)\bigr\rvert \\\label{shockwavebound}
										&=\Bigl\lvert \frac{1}{N} \sum\limits_{i=1}^N h(t,x - x_i, \xi_i)  - \int h(t, x-y, \xi)\, f_0(y,\xi) \Bigr \rvert,
									\end{align}
									\noindent which depends only on initial data. Applying Proposition \ref{Prop:LLN} (with $\phi(z)=z$ and $\alpha = 3 \gamma, \, \beta = \frac{1}{4}$) we have for any $(t,x)$: 
									\begin{equation*} \IP_0\Bigl[ \Bigl\lvert \frac{1}{N} \sum\limits_{i=1}^N h(t,x - x_i, \xi_i)  - \int h(t, x-y, \xi)\, f_0(y,\xi) \Bigr \rvert > N^{-\frac{1}{4}} \Bigr] \lesssim N^{-4} \end{equation*}
									and thus $\IP_0\bigr[\exists x \in \mathcal{G}^N \mid \eqref{shockwavebound} > N^{-\frac{1}{4}} \bigr]\lesssim N^{-1}$. We conclude:\\
									
									\noindent Let $A^5_t$ be the ($N$ and $t$ dependent) set defined by
									\begin{equation}\label{A5} A^5_t := \lbrace Z \in \IR^3\times\IR^3 \mid \eqref{shockwavebound} \leq N^{-\frac{1}{4}}  \text{ for all } x \in \mathcal{G}^N  \rbrace.\end{equation}
									Then there exists $C_5 > 0$ such that $\IP(A^5_t) \geq 1- \frac{C_5}{N}$.\\
									
									\noindent \textbf{Remark:} Without regularization, the kernel \eqref{hshocks} would have the form $t \int\limits_{S^2} \frac{\omega - v}{1 - v\cdot \omega} \, \delta(x-wt)$, which is not only unbounded, but distribution valued, reflecting the fact that $E'_0(t,x)$ depends on the initial charge distribution only via $\rho_0 \bigl\lvert_{\lbrace \lvert x - y \rvert = t\rbrace}$. However, after smearing with $\chi^N$, the term is relatively harmless. The width of the necessary cut-off for the law of large number estimate could be further reduced by exploiting the fact that $h(t,x,\xi) = 0$ unless $t - r_n < \lvert x \rvert  < t +r_N $.\\
									
									\noindent For the magnetic field component $B'_0$, the proof works analogously, since the corresponding kernel satisfies the same bound \eqref{swkernelbound}.
									
									\subsection{Controlling the homogeneous fields}
									It remains to control the contribution of the homogeneous fields \eqref{E0term}, which depend only on the initial data via the Gauss constraint $\Div {E_0} \bigl\lvert_{t=0} = \rho_0$. The solution of the homogeneous field-equation is given by
									
									\begin{equation*} E_0(t,x) = \partial_t Y(t, \cdot) *E_{in} (x) = \partial _t \bigl(\frac{t}{4\pi} \int\limits_{S^2} E_{in}( x + \omega t) \dd \omega \bigr). \end{equation*}
									
									\noindent If $E_{in}(x)=- \nabla G * \rho_0(x) = \int \frac{x-y}{\lvert x -y\rvert^3}\rho_0(y) \, \dd y $ is the Coulomb field, we compute: 
									\begin{align*}
										&- \partial_t \nabla_x \int G *_x Y(t, \cdot) *_x \tilde f_0(x, \xi) \dd \xi\\
										& = \frac{1}{4\pi}\int \int\limits_{S^2}\Bigl[ \frac{x-y + 2 \omega t}{\lvert x-y + \omega t \rvert^3} -  \frac{t \omega \cdot (x-y + \omega t) (x-y + \omega t) }{\lvert x-y + \omega t \rvert^5} \Bigr] \dd \omega \, \tilde \rho_0(y)\,  \dd y\\
										& = \frac{1}{4\pi} \int \int\limits_{S^2} h'(t \omega,  x- y ) \, \dd \omega \, \tilde \rho_0(y) \dd y,
									\end{align*}
									with $h'(t\omega, x):= \frac{1}{4\pi} \bigl(\frac{x+ 2 \omega t}{\lvert x + \omega t \rvert^3} -  \frac{t \omega \cdot (x + \omega t) (x+ \omega t) }{\lvert x + \omega t \rvert^5}\bigr)$. Shifting the mollifier to the kernel, we get:
									\begin{equation*}\lvert \chi^N * h' \rvert \lesssim  r_N^{-2} + t \, r_N^{-3}, \end{equation*}
									where we used again Lemma \ref{chibounds}, and thus 
									\begin{equation}\label{homogbound} E_0(t,x) = \int \int {h_0(t, x-y) \, f_0(y, \xi)} \, \dd y \dd \xi,\end{equation}
									with 
									\begin{equation} 
										h_0(t,x)  := \int\limits_{S_1} \chi^N*h'(x, \omega t) \dd \omega, \;\;\; \lvert h_0(t,x) \rvert   \lesssim r_N^{-2} + t \, r_N^{-3}.  \end{equation}
									Now, by \eqref{macroin}, the incoming fields are fixed such that $E^N_{in} - E^\mu_{in} = - \nabla G * (\rho_0[f] - \rho_0[\mu[Z]])$. Hence, we have to control the difference
									\begin{equation}\label{Einbound} 
										\Bigl\lvert \frac{1}{N} \sum\limits_{i=1}^N h_0(t,x - x_i)  - \int h_0(t, x-y)\, f_0(y,\xi) \dd y \dd \xi \Bigr \rvert.
									\end{equation}
									
									\noindent As before, an application of the law of large numbers in form of Proposition \ref{Prop:LLN} yields the following: Let $A^6_t$ be the ($N$ and $t$ dependent) set defined by
									\begin{equation}\label{A6} A^6_t := \lbrace Z \in \IR^3\times\IR^3 \mid \eqref{homogbound} \leq N^{-\frac{1}{4}} \text{ for all } x \in \mathcal{G}^N\rbrace. \end{equation}
									Then there exists $C_6 > 0$ such that $\IP_0(A^6_t) \geq 1- \frac{C_6}{N}$.\\
									
									\noindent For the magnetic field, $B^N_0 - B^\mu_0 = 0$ since, by assumption, $B^N_{in} = B^\mu_{in}$. \\

									\noindent For every $t$, our law of large numbers estimates yield bounds on a finite number on points, that we have chosen to lie on the grid $\mathcal{G}^N$ covering the interval $[-\overline{r}, \overline{r}]$ which contains the support of $f^N$ and $\mu^N$. However, combined with the bound on the field derivatives from Proposition \ref{Prop:deribound}, this can be used to derive a $L^\infty$-bound. We give an example in the following lemma. 
									
									\begin{Lemma}\label{Lemma:Einbound}
										Let $\overline{r}$ as defined in \eqref{xsupport}. In view of the assumptions of Propositions \ref{Prop:rhobound} and \ref{Prop:deribound}, we fix some $p \geq 1$ and consider the set $M = M(p)$ defined by 
										\begin{equation} Z \in  M \iff W^p_p(\mu^N_0[Z], f_0) \leq r_N^{3+p}. \end{equation} 
										Let $E^N_{in}$ and $E^\mu_{in} =E^\mu_{in}[Z]$ as fixed in \eqref{microin}. Then there exists a constant $C >0$ such that
										
										\begin{equation}\IP_0 \Bigl[ \lVert E^N_{in} - E^\mu_{in} \rVert_{L^\infty(\mathrm{B}(\overline r))} \lesssim N^{-\frac{1}{4}}\, \Bigr\rvert \, M \Bigr]  \geq 1 - \frac{C}{N}.\end{equation}
									\end{Lemma}
									\begin{proof} 
										Above, we have proven that 
										\begin{equation} \IP_0\Bigl[ \exists x_k \in \mathcal{G}^N : \lvert E^N_{in}(x_k) - E^\mu_{in}(x_k) \rvert \geq  N^{-\frac{1}{4}}\Bigr] \lesssim N^{-1}. \end{equation} 
										\noindent Furthermore, according to Proposition \ref{Prop:deribound}, we have $\Vert \nabla_x (E^N - E^\mu) \rVert_{\infty} \lesssim r_N^{-2}$ for $Z \in M$. By construction: 
										$ \sup \lbrace \min \limits_{x_i \in \mathcal{G}^N} \lvert x -  x_i \rvert : x \in \mathrm{B(\overline{r})} \rbrace \leq \frac{\sqrt{3}}{2} \frac{\overline{r}}{N} $.
										Hence, $ \lvert E^N_{in}(x_k) - E^\mu_{in}(x_k) \rvert \leq  N^{-\frac{1}{4}} \; \forall x_k  \in \mathcal{G} $ implies $ \lvert E^N_{in}(x) - E^\mu_{in}(x) \rvert \lesssim N^{-\frac{1}{4}} +  \frac{r_N^{-2}}{N}  \leq  N^{-\frac{1}{4}} + N^{-1 + 2\gamma} $ for all $x \in \mathrm{B(\overline{r})}$. Since $\gamma < \frac{1}{12}$, we conclude
										\begin{equation*}\IP_0 \Bigl[ \lVert E^N_{in} - E^\mu_{in} \rVert_{L^\infty(\mathrm{B}(\overline r))} \lesssim N^{-\frac{1}{4}} \Bigr\rvert Z \in  M \Bigr] \lesssim N^{-1}.\end{equation*}
									\end{proof} 
									\section{A Gronwall argument}\label{section:Gronwall}
									We are finally ready to combine the results of the previous sections into a proof of the main theorem. Our aim is to establish a Gronwall bound for the quantity $\IE_0(J_t^N)$ defined in \ref{Def:J2}, thus proving the mean field limit for typical initial conditions.\\
									
									\noindent In order to control the evolution of  $J^{N}_t(Z)$ we will need the following Lemma.
									
									\begin{Lemma}\label{Lemma:semiderivative}
										For a function $g: \IR \to \IR$, we denote by
										\begin{equation} 
										\partial_t^+ g(t) := \lim\limits_{\Delta t \searrow 0} \frac{g(t + \Delta t) - g(t)}{\Delta t} 
										\end{equation}
										the right-derivative of $f$ with respect to $t$. Let $g \in C^1(\IR)$ and $h(t):= \sup\limits_{0 \leq s \leq t} g(t)$. Then $\partial_t^+ h(t)$ exists and $\partial_t^+ h(t) \leq \min \lbrace 0 , g'(t) \rbrace$ for all $t$. 
									\end{Lemma}
									
									\subsection{Good initial conditions}  Let $\gamma < \frac{1}{12}$ and $r_N \geq N^{-\gamma}$. Fix an initial distribution $f_0$ with compact support as in Theorem \ref{Thm:Thm3}. We begin by noting the (time-independent) conditions that the initial configuration $Z \in \IR^{6N}$ has to satisfy. All probabilities are meant with respect to the product-measure $\otimes^N f_0$ on  $\IR^{3N}$. Consider the sets $\mathrm{C}_1 , \mathrm{C}_2$ defined by
									\begin{align}
										&Z \in \mathrm{C}_1 \iff z_i \in \supp (f_0) , \forall 1 \leq i \leq N. \\\label{C2}
										&Z \in  \mathrm{C}_2 \iff  \lVert (E^N_{in},B^N_{in}) - (E^\mu_{in},B^\mu_{in})  \rVert_{L^\infty(\mathrm{B}(\overline r))}\leq  N^{-\frac{1}{4}}.
									\end{align}
									Moreover, setting $p := \frac{1}{4\gamma}$, we consider the set  $\mathrm{C}_3 \subset \IR^{6N}$ defined by
									\begin{align}
										&Z \in \mathrm{C}_3 \iff  W^p_p(\mu^N[Z], f_0) \leq r_N^{3+p}. \hspace{4.1cm}
									\end{align}
									Obviously, $\IP_0( Z \notin  \mathrm{C}_1 ) = 0$ and according to Lemma  \ref{Lemma:Einbound}, $\IP_0( Z \notin  \mathrm{C}_2 ) \lesssim N^{-1}$. For $ \mathrm{C}_3 $, we apply the large deviation estimate, Theorem \ref{Fournier}, with $d=6$, $p := \frac{1}{4\gamma}$ and  $\xi = r_N^{3+p} \geq N^{- (3+p)\gamma} = N^{-(3\gamma + 1/4)}$. This yields constants $c,c' > 0$ such that
									\begin{equation} \IP_0 \Bigl( W^p_p(\mu^N_0[Z], f_0) > r_N^{3 + p} \Bigr) \leq  c' e^{-c N^s},\end{equation}
									where 
									\begin{equation}\label{sofgamma} s = 1 - 2(3\gamma + {1}/{4}) = \frac{1}{2} ( 1 - {12}\gamma ) > 0.\end{equation} 
									
									\noindent In total, setting
									\begin{equation} \mathcal{C}:= \mathrm{C}_1 \cap \mathrm{C}_2 \cap \mathrm{C}_3, \end{equation}
									there exists a constant $C_7$ such that $\IP_0( \mathcal{C}) \geq 1 - \frac{C_7}{N}$. Note that the requirement $\gamma < \frac{1}{12}$ for the width of the cut-off comes from \eqref{sofgamma}.
									
									\subsection{Evolution of $J^N_t$} For $t >0$ we have to control the growth of $\IE_0(J^N_t)$.
									Recall from Def. \ref{Def:J2}: 
									\begin{equation*}\begin{split}J^N_t(Z) := \min \Bigl\lbrace 1, \lambda(N) N^{\delta}\sup\limits_{0 \leq s \leq t} \lvert ^N\Psi^1_{t,0}(Z) - {}^N\Phi^1_{t,0}(Z) \rvert_\infty&\\ +  N^\delta \sup\limits_{0 \leq s \leq t} \lvert ^N\Psi^2_{t,0}(Z) - {}^N\Phi^2_{t,0}(Z) \rvert_\infty &\Bigr\rbrace, \end{split} \end{equation*}
									with $\lambda(N) := \max \lbrace 1, \sqrt{\log(N)} \rbrace.$ For fixed $t >0$ we denote by $\mathcal{B}_t$ the set
									\begin{equation} \mathcal{B}_t := \lbrace Z \in \IR^3 \times \IR^3 : J^N_t (Z) < 1 \rbrace. \end{equation}
									
									\noindent Moreover, we define the set 
									\begin{equation}
										\mathcal{A}_t := A^1_t \cap A^2_t \cap A^3_t \cap A^4_t \cap ... \cap A^{12}_{t},
									\end{equation}
									where $A^1_t, A^2_t, A^3_t, A^4_t, A^5_t, A^6_t$ are defined in Section \ref{section:pointwiseestimates} and $A_t^7,.., A^{12}_t$ are the analogous sets for the magnetic field components.\\
									
									\noindent We  split $\IE_0(J^N_t)$ into 
									\begin{align*}
										\IE_0(J^N_t) &= \IE_0(J^N_t \mid \mathcal{A}_t\cap \mathcal{B}_t \cap \mathcal{C}) +  \IE_0(J^N_t \mid \mathcal{B}_t \cap (\mathcal{A}_t \cap \mathcal{C})^c) +  \IE_0(J^N_t \mid \mathcal{B}_t^c ).
									\end{align*}
									\noindent Now, we first observe that if  $Z \in \mathcal{B}_t^c$, we have $\frac{\dd}{\dd t} J_t^N = 0$, since $J_t^N(Z)=1$ is already maximal. In particular, 
									\begin{equation} \partial_t \,\IE_0(J^N_t \mid \mathcal{B}_t^c) = 0.
									\end{equation}  
									
									\noindent Hence, we only need to consider the case $J^N_t(Z) < 1 $ for which, in particular,
									\begin{equation}\label{close} \sup\limits_{0 \leq s \leq t} \lvert {}^{N}\Psi_{s,0}(Z) - {}^{N}\Phi_{s,0}(Z) \rvert_\infty < N^{-\delta} \leq N^{-\gamma} \leq r_N. \end{equation} 
									We have to control the evolution of
									\begin{align*}
										\lambda(N) N^{\delta}\sup\limits_{0 \leq s \leq t} \lvert {}^{N}\Psi^1_{s,0}(Z) - {}^{N}\Phi^1_{s,0}(Z) \rvert_\infty +  N^\delta \sup\limits_{0 \leq s \leq t} \lvert {}^{N}\Psi^2_{s,0}(Z) - {}^{N}\Phi^2_{s,0}(Z) \rvert_\infty.
									\end{align*}
									
									\noindent We will denote by $E^N = E^N [\tilde f^N]$ and $B^N = B^N[\tilde f^N]$ the macroscopic fields, generated by the (regularized) Vlasov density, and by $E^\mu = E^\mu[\tilde \mu^N[Z]], B^\mu = B^\mu[\tilde \mu^N[Z]]$ the microscopic fields, generated by the rigid charges.\\
									
									\noindent Recalling Lemma \ref{Lemma:semiderivative} and denoting by $\partial_t^+$ the derivative from the right w.r.t. $t$, we find:
									\begin{equation}\begin{split}\label{controlQ2}  \partial_t^+ &\sup_{0 \leq s \leq t} \lvert {}^{N}\Psi^1_{s,0}(Z) - {}^{N}\Phi^1_{s,0}(Z) \rvert_\infty \\
											& \leq \bigl\lvert \partial_t  ({}^{N}\Psi^1_{t,0}(Z) - {}^{N}\Phi^1_{t,0}(Z) ) \bigr\rvert_\infty 
											=  \max_{1\leq i \leq N} \lvert v(\xi^*_i(t)) - v(\eta^*_i(t))\rvert \\ &\leq 2 \max_{1\leq i \leq N} \lvert \xi^*_i(t) - \eta^*_i(t)\rvert  =  2 \lvert {}^{N}\Psi^2_{t,0}(Z) - {}^{N}\Phi^2_{t,0}(Z) \rvert_\infty,\end{split}\end{equation}
									\noindent as well as 
									\begin{equation}\label{Pevolution} \begin{split} 
											\partial_t^+ & \sup_{0 \leq s \leq t} \lvert {}^{N}\Psi^2_{s,0}(Z) - {}^{N}\Phi^2_{s,0}(Z) \rvert_\infty \\
											\leq &\bigl\lvert \partial_t  ({}^{N}\Psi^2_{t,0}(Z) - {}^{N}\Phi^2_{t,0}(Z) ) \bigr\rvert_\infty =   \max_{1\leq i \leq N} \lvert \tilde K[\tilde \mu](t, x^*_i, \xi^*_i) - \tilde{K}[\tilde f](t, y^*_i, \eta^*_i) \rvert\\  
											\leq &  \max_{1\leq i \leq N} \lvert \tilde K[\tilde f](t, x^*_i, \xi^*_i) - \tilde{K}[\tilde f](t, y^*_i, \eta^*_i) \rvert +   \max_{1\leq i \leq N} \lvert \tilde K[\tilde\mu](t, y^*_i, \eta^*_i) - \tilde{K}[\tilde f](t, y^*_i, \eta^*_i) \rvert \\
											\leq & L  \lvert {}^{N}\Psi _{t,0}(Z) - {}^{N}\Phi_{t,0}(Z) \rvert_\infty + \lVert \tilde E^N(t) -  \tilde E^\mu(t) \rVert_{L^\infty(\mathrm{B}(\overline{r}))} + \lVert \tilde B^N(t) -  \tilde B^\mu(t) \rVert_{L^\infty(\mathrm{B}(\overline{r}))}
										\end{split}\end{equation} 
										In the last line, we used the uniform Lipschitz bound on the mean field force \eqref{LipschitzL} and the fact that $\lvert x^*_i\rvert, \lvert y^*_i\rvert < \overline{r}$ for all $i=1,..,N$ and $t \leq T$.\\
										
										\noindent  It remains to control the term
										\begin{equation} \begin{split} \lVert \tilde E^N(t,\cdot) -  \tilde E^\mu(t,\cdot) \rVert_{L^\infty(\mathrm{B}(\overline{r}))} + \lVert \tilde B^N(t,\cdot) -  \tilde B^\mu(t,\cdot) \rVert_{L^\infty(\mathrm{B}(\overline{r}))}\\
												\leq  \lVert E^N(t,\cdot) -  E^\mu(t,\cdot) \rVert_{L^\infty(\mathrm{B}(\overline{r}))} + \lVert B^N(t,\cdot) -  B^\mu(t,\cdot) \rVert_{L^\infty(\mathrm{B}(\overline{r}))}.\end{split}\end{equation}

										\noindent Now, $Z \in (\mathcal{A}_t \cap \mathcal{C})^c$ are the ``bad'' initial conditions that may lead to large fluctuations in the fields or a blow-up of the microscopic charge density. However, the Vlasov fields $(\tilde E^N, \tilde B^N)$ are bounded uniformly in $N$ according to \eqref{LipschitzL}, while the (smeared) microscopic fields $(\tilde E^\mu , \tilde B^\mu)$ diverge at most as $\Vert (\tilde E^\mu ,  \tilde B^\mu) \rVert_\infty \lesssim r_N^{-2}$ according to Prop. \ref{Prop:deribound}. Therefore:
										\begin{equation} \begin{split}
												&\lVert \partial_t^+ J^N_t (\cdot) \rVert_{L^\infty(\IR^{6N})}\\ 
												&\leq (2\lambda(N) + L) J^N_t  + \lVert \tilde E^N_t \rVert_\infty +  \lVert \tilde E^\mu_t \rVert_{\infty} + \lVert \tilde B^N_t \rVert_\infty +  \lVert  \tilde B^\mu_t \rVert_\infty \lesssim r_N^{-2}.
											\end{split}\end{equation} 
											Hence, there exists a constant $C'$ such that 
											\begin{equation}\begin{split}\label{atypicalbound}
													&\partial_t^+  \IE_0(J^N_t \mid \mathcal{B}_t \cap (\mathcal{A}_t \cap \mathcal{C})^c)) =  \IE_0(\partial_t^+ J^N_t \mid \mathcal{B}_t \cap (\mathcal{A}_t \cap \mathcal{C})^c) \\
													&\leq \lVert \partial_t^+ J^N_t \rVert_{L^\infty(\IR^{6N})} \,  \IP_0( \mathcal{A}_t^c \cup \mathcal{C}^c) \leq C' r_N^{-2} \frac{1}{N} \leq C' N^{-1 + 2 \gamma}.
												\end{split}\end{equation}

												\noindent $Z \in \mathcal{A}_t \cap \mathcal{B}_t \cap \mathcal{C}$ are the ``good'' initial conditions, for which we have derived various nice properties:
												\begin{itemize}
													\item[] $\lvert x^*_i(t) \rvert < \overline{r}, \, \lvert \xi^*_i(t) \rvert < \overline{\xi}, \,\forall  t\in [0,T]$ \hfill (from eq. \ref{close})
													\item[] $\lVert \rho[\mu^N_t[Z]]\rVert_\infty \leq C_\rho,\, \forall N\geq 1, t \in [0,T]$ \hfill (from Proposition \ref{Prop:rhobound})
													\item[] $\lVert (\nabla_x E^\mu, \nabla_x B^\mu) \rVert_\infty \lesssim r_N^{-2}$ \hfill (Proposition \ref{Prop:deribound})
													\item[] $\lVert (E^N_{in},B^N_{in}) - (E^\mu_{in},B^\mu_{in}) \rVert_{L^\infty(\mathrm{B}(\overline r))}\leq  N^{-1/4}$ \hfill (since $Z \in  \mathrm{C}_2$)\\
												\end{itemize}
												\noindent In particular, combining the results of Section \ref{section:pointwiseestimates}, we have:
												\begin{align*} &\max \bigl\lbrace \lvert  E^N (t, x_i) -  E^\mu(t, x_i) \rvert_\infty +  \lvert  B^N (t, x_i) -  B^\mu(t, x_i) \rvert_\infty : x_i \in \mathcal{G}^N \bigr\rbrace\\
													& \lesssim \;  \underbrace{\; \frac{}{} N^{-\frac{1}{4}} \;}_{from \, (\ref{A1}, \ref{A3}, \ref{A5}, \ref{A6})} + \underbrace{\frac{C_0}{(1-\overline{v})^4} ( 1 + \log(r^{-1}_N)) \sup\limits_{0 \leq s \leq t} \lvert {}^{N}\Phi^1_{s,0}(Z) - {}^{N}\Psi^1_{s,0}(Z) \rvert_\infty}_{from \,(\ref{gterm1},\ref{gterm2} \,   , \ref{A4})}\\
													& + \underbrace{ \frac{L C_0 T}{(1-\overline{v})^5} \sup\limits_{0 \leq s \leq t} \lvert {}^{N}\Phi^1_{s,0}(Z) - {}^{N}\Psi^1_{s,0}(Z) \rvert_\infty + \frac{ C_0 T}{(1-\overline{v})^4} \sup\limits_{0 \leq s \leq t} \lvert {}^{N}\Phi^2_{s,0}(Z) - {}^{N}\Psi^2_{s,0}(Z) \rvert_\infty}_{from (\ref{g2term1}, \ref{g2term2} \,   , \ref{A4})}\\
													& + \underbrace{ \frac{L C_0 (1 + T^2)}{(1-\overline{v})^4} \Bigl(\sup\limits_{0 \leq s \leq t} \lvert {}^{N}\Phi^1_{s,0}(Z) - {}^{N}\Psi^1_{s,0}(Z) \rvert_\infty + \sup\limits_{0 \leq s \leq t} \lvert {}^{N}\Phi^2_{s,0}(Z) - {}^{N}\Psi^2_{s,0}(Z) \rvert_\infty\Bigr)}_{from (\ref{g3term1} - \ref{g1term2})}\\
													& +  \underbrace{\frac{C_\rho}{(1-\overline{v})^2} \int\limits^t_0  \lVert E^N (s) -  E^\mu(s) \rVert_{L^\infty(\mathrm{B(\overline{r})})} +  \lVert  B^N (s) -  B^\mu(s) \rVert_{L^\infty(\mathrm{B(\overline{r})})} \, \dd s. }_{from \, \eqref{KGronwall}}
												\end{align*}
												We simplify this expression to:
												\begin{equation}\label{simplified}\begin{split}
													&\max \bigl\lbrace \lvert  E^N (t, x_i) -  E^\mu(t, x_i) \rvert_\infty +  \lvert  B^N (t, x_i) -  B^\mu(t, x_i) \rvert_\infty : x_i \in \mathcal{G}^N \bigr\rbrace\\
													&\lesssim  \; N^{-\frac{1}{4}} + \frac{C_0  \log(r^{-1}_N)}{(1-\overline{v})^4} \sup\limits_{0 \leq s \leq t} \lvert {}^{N}\Phi^1_{s,0}(Z) - {}^{N}\Psi^1_{s,0}(Z) \rvert_\infty\\
													& + \frac{L C_0 (1 + T^2)}{(1-\overline{v})^5} \sup\limits_{0 \leq s \leq t} \lvert {}^{N}\Phi_{s,0}(Z) - {}^{N}\Psi_{s,0}(Z) \rvert_\infty\\ 
													& +  \frac{C_\rho}{(1-\overline{v})^2} \int\limits^t_0  \lVert E^N (s) -  E^\mu(s) \rVert_{L^\infty(\mathrm{B(\overline{r})})} +  \lVert  B^N (s) -  B^\mu(s) \rVert_{L^\infty(\mathrm{B(\overline{r})})} \, \dd s.
													\end{split}
												\end{equation}

												\noindent According to Proposition \ref{Prop:deribound} and equation \eqref{LipschitzL}, we have $\Vert (E^N, B^N) - (E^\mu, B^\mu) \rVert_{Lip} \lesssim r_N^{-2}$. Moreover, by construction: 
												$ \sup \bigl\lbrace \min \limits_{x_i \in \mathcal{G}^N} \lvert x -  x_i \rvert : x \in \mathrm{B(\overline{r}, 0)} \bigr\rbrace \leq \frac{\sqrt{3}}{2} \frac{\overline{r}}{N}.$ Hence, by the same argument as in Lemma \ref{Lemma:Einbound},
												\begin{align*}\notag  &\lVert  E^N (t, \cdot) -  E^\mu (t, \cdot) \rVert_{L^\infty(\mathrm{B(\overline{r})})} +   \lVert  B^N (t, \cdot) -  B^\mu (t, \cdot) \rVert_{L^\infty(\mathrm{B(\overline{r})})} \\\notag
													&\lesssim \max \bigl\lbrace \lvert  E^N (t, x_i) -  E^\mu(t, x_i) \rvert_\infty +  \lvert  B^N (t, x_i) -  B^\mu(t, x_i) \rvert_\infty : x_i \in \mathcal{G}^N \bigr\rbrace + \frac{r_N^{-2}}{N}, \end{align*}
												
												\noindent where $\frac{r_N^{-2}}{N} \leq N^{-1+2\gamma} \leq N^{-\frac{1}{4}}$. Together with \eqref{simplified}, we thus have:
												
												\begin{align*} 
													&\lVert  E^N (t, \cdot) -  E^\mu (t, \cdot) \rVert_{L^\infty(\mathrm{B(\overline{r})})} +   \lVert  B^N (t, \cdot) -  B^\mu (t, \cdot) \rVert_{L^\infty(\mathrm{B(\overline{r})})}\\[1.5ex]
													& \lesssim  N^{-\frac{1}{4}} + \frac{C_0  \log(r^{-1}_N)}{(1-\overline{v})^4} \sup\limits_{0 \leq s \leq t} \lvert {}^{N}\Phi^1_{s,0}(Z) - {}^{N}\Psi^1_{s,0}(Z) \rvert_\infty\\
													& + \frac{L C_0 (1 + T^2)}{(1-\overline{v})^5} \sup\limits_{0 \leq s \leq t} \lvert {}^{N}\Phi_{s,0}(Z) - {}^{N}\Psi_{s,0}(Z) \rvert_\infty\\
													& +  \frac{C_\rho}{(1-\overline{v})^2} \int\limits^t_0  \lVert E^N (s) -  E^\mu(s) \rVert_{L^\infty(\mathrm{B(\overline{r})})} +  \lVert  B^N (s) -  B^\mu(s) \rVert_{L^\infty(\mathrm{B(\overline{r})})} \, \dd s. \end{align*}                                                                                                                                                                                                                                                                                                                                                                                                                                                                                                                                                                                                                                                                                                                                   
												
												\noindent By Gronwall's inequality, there exists a constant $C'' > 0$ depending on $\overline{v}$ and $C_\rho$ such that
												
												\begin{equation*}\begin{split}&\lVert  E^N (t, \cdot) -  E^\mu (t, \cdot) \rVert_{L^\infty(\mathrm{B(\overline{r})})}+   \lVert  B^N (t, \cdot) -  B^\mu (t, \cdot) \rVert_{L^\infty(\mathrm{B(\overline{r})})}\\[1.5ex]
														& \leq e^{tC''} \Bigl( N^{-\frac{1}{4}} + \frac{ C_0\log(r^{-1}_N)}{(1-\overline{v})^4} \sup\limits_{0 \leq s \leq t} \lvert {}^{N}\Phi^1_{s,0}(Z) - {}^{N}\Psi^1_{s,0}(Z) \rvert_\infty\\
														& \hspace{15mm} + \frac{LC_0(1 + T^2)}{(1-\overline{v})^5} \sup\limits_{0 \leq s \leq t} \lvert {}^{N}\Phi_{s,0}(Z) - {}^{N}\Psi_{s,0}(Z) \rvert_\infty\\
														&\hspace{15mm}  + \hfill \lVert E^N(0 , \cdot) - E^\mu(0, \cdot) \rVert_{L^\infty(\mathrm{B(\overline{r})})}+ \lVert B^N(0 , \cdot) - B^\mu(0, \cdot) \rVert_{L^\infty(\mathrm{B(\overline{r})})} \Bigr) \end{split}\end{equation*}
												\noindent and with \eqref{C2}:
							\begin{equation}\label{Kresult}\begin{split}
														&\lVert E^N(t,\cdot) -  E^\mu(t,\cdot) \rVert_{L^\infty(\mathrm{B}(\overline{r}))} + \lVert B^N(t,\cdot) -  B^\mu(t,\cdot) \rVert_{L^\infty(\mathrm{B}(\overline{r}))} \\[1.5ex] 
														&\leq e^{TC''} \frac{C_0\log(r^{-1}_N)}{(1-\overline{v})^4} \sup\limits_{0 \leq s \leq t} \lvert {}^{N}\Phi^1_{s,0}(Z) - {}^{N}\Psi^1_{s,0}(Z) \rvert_\infty\\
														& + e^{TC''}  \frac{LC_0(1+T^2)}{(1-\overline{v})^5} \sup\limits_{0 \leq s \leq t} \lvert {}^{N}\Phi_{s,0}(Z) - {}^{N}\Psi_{s,0}(Z)\rvert_\infty + e^{TC''}  2N^{-\frac{1}{4}} .\end{split}\end{equation}
												
												\noindent Plugging this into \eqref{Pevolution}, we get:
												\begin{equation}\begin{split}
														\partial_t^+ \,&\bigl( N^\delta \sup\limits_{0 \leq s \leq t} \lvert {}^{N}\Psi^2_{s,0}(Z) - {}^{N}\Phi^2_{s,0}(Z) \rvert_\infty\bigr)\\
														&\leq  N^\delta L \lvert {}^{N}\Psi_{t,0}(Z) - {}^{N}\Phi_{t,0}(Z) \rvert_\infty +  2 e^{TC''} N^{-\frac{1}{4} + \delta}\\
														& +e^{TC''}  \frac{LC_0(1+T^2)}{(1-\overline{v})^5} N^\delta \sup\limits_{0 \leq s \leq t} \lvert {}^{N}\Phi_{s,0}(Z) - {}^{N}\Psi_{s,0}(Z) \rvert_\infty\\
														& + e^{TC''} \frac{C_0\log(N)}{(1-\overline{v})^4} \,  N^\delta  \sup\limits_{0 \leq s \leq t} \lvert {}^{N}\Phi^1_{s,0}(Z) - {}^{N}\Psi^1_{s,0}(Z)\rvert_\infty.
													\end{split}\end{equation}
													
													\noindent Note, in particular, that the last summand can be rewritten as: \begin{align*}
														\frac{ \sqrt{\log(N)}}{(1-\overline{v})^4}  \Bigl( \sqrt{\log(N)} N^\delta  \sup\limits_{0 \leq s \leq t} \lvert {}^{N}\Phi^1_{s,0}(Z) - {}^{N}\Psi^1_{s,0}(Z) \rvert_\infty\Bigr),
													\end{align*}
													so that, together with \eqref{controlQ2} and $\lambda(N) = \max \lbrace 1, \sqrt{\log(N)} \rbrace$:
													\begin{align*}
														\partial_t^+ J^N_t(Z) &\leq  2 \lambda(N) N^\delta \lvert {}^{N}\Psi^2_{t,0}(Z) - {}^{N}\Phi^2_{t,0}(Z) \rvert_\infty + N^\delta L \lvert {}^{N}\Psi_{t,0}(Z) - {}^{N}\Phi_{t,0}(Z) \rvert_\infty \\
														&+ e^{TC''}\frac{ C_0\sqrt{\log(N)}}{(1-\overline{v})^4}  \Bigl( \sqrt{\log(N)} N^\delta  \sup\limits_{0 \leq s \leq t} \lvert {}^{N}\Psi^1_{s,0}(Z) - {}^{N}\Phi^1_{s,0}(Z) \rvert_\infty\Bigr)\\
														&+ e^{TC''} \frac{LC_0(1+T^2)}{(1-\overline{v})^5} N^\delta \sup\limits_{0 \leq s \leq t} \lvert {}^{N}\Phi_{s,0}(Z) - {}^{N}\Psi_{s,0}(Z) \rvert_\infty +  2e^{TC''}N^{-\frac{1}{4} + \delta} \\ 
														& \leq \frac{e^{TC''}LC_0(3+T^2)}{(1-\overline{v})^5} \lambda(N)\, J^N_t(Z)  + 2e^{TC''}N^{-\frac{1}{4} + \delta} .
													\end{align*}
													
													\noindent Together with \eqref{atypicalbound}, we have found:
													\begin{align*}
														\IE_0(&J^N_{t + \Delta t}) - \IE_0(J^N_{t,0}) \\
														&\leq  \Bigl(\frac{e^{TC''}LC_0(3+T^2)}{(1-\overline{v})^5} \lambda(N)\, J^N_t(Z)  + 2e^{TC''}N^{-\frac{1}{4} + \delta} +  C' N^{-1+2\delta} \Bigr) \Delta t + o(\Delta t).
													\end{align*}
													Finally, using Gronwalls inequality and the fact that $J_0^N(Z)=0\; \forall Z$ we get
													\begin{equation} \IE_0(J^N_t) \leq {e^{tC\lambda(N)}} N^{-\frac{1}{4} + \delta}, \end{equation}
													with \begin{equation}
														C(T, C_0, f_0) =  \max \Bigl\lbrace \frac{e^{TC''}LC_0(3+T^2)}{(1-\overline{v})^5},\,  C' \Bigr\rbrace.
													\end{equation} 
													Together with the results of Section \ref{Section:strategyofproof}, Proposition \ref{Prop:mufromJ} and Lemma \ref{Cor:largedeviation}, this concludes the proof of the theorem. For simplicity, we demand $N \geq 4$ so that $\lambda(N)= \sqrt{\log(N)}$.\\

													\noindent The approximation result for the fields, i.e. part c) of the theorem, can be read off equation \eqref{Kresult} using $\IP_0\bigl[ \sqrt{\log(N)}\sup_{0 \leq s \leq t} \lvert {}^N\Phi^1_{s,0}(Z) - {}^N\Psi^1_{s,0}(Z) \rvert_\infty \geq N^{-\delta}\bigr] \leq \IE_0(J^N_t)$ as well as $\IP_0\bigl[\sup_{0 \leq s \leq t} \, \lvert {}^N\Phi_{s,0}(Z) - {}^N\Psi_{s,0}(Z) \rvert_\infty \geq N^{-\delta}\bigr] \leq \IE_0(J^N_t)$. By choosing the grid $\mathcal{G}^N$ accordingly, $\mathrm{B}(\overline{r})$ can be replaced by any compact set $M \subset \IR^3$.\\
													\qed

\section{Concluding remarks}\label{section:remarks}
	We have presented a particle approximation for the Vlasov-Maxwell dynamics that improves significantly on previous results, allowing generic initial data (for the particles) and an $N$-dependent cut-off decreasing as  $N^{-\delta}$ with $\delta < \frac{1}{12}$. Still, our derivation leaves much room for improvement as far as the size of the cut-off is concerned. Note that the restriction $\delta < \frac{1}{12}$ comes only from the Wasserstein bound on the charge density, Prop. \ref{Prop:rhobound}, which assures that the microscopic charge density typically remains bounded uniformly in $N$ and $t$. This is a relatively powerful but rather coarse way to prevent a blow-up of the microscopic dynamics. All other estimates would allow the cut-off (electron radius) to decrease at least with $\delta < \frac{1}{4}$, even with the rough law of large number estimates used in Section \ref{section:pointwiseestimates}. Hence, it seems likely that the width of the cut-off could be further reduced by a more detailed analysis of the microscopic dynamics, in particular the ``acceleration / radiation'' component of the electromagnetic field. 
	
	However, regarding the rigid charges model considered here, we want to emphasize that the status of the regularization is different in the context of Vlasov-Maxwell than with respect to the nonrelativistic Coulomb interactions considered in \cite{Dustin},\cite{PeterDustin}. In the latter case, the correct microscopic dynamics are known and quite well understood. Any regularization thereof is first and foremost a simplification of the mathematical problem with the width of the cut-off essentially quantifying the deviation from the true microscopic theory. When it comes to the relativistic regime, however, the standard Maxwell-Lorentz dynamics are {not} well defined for point-particles due to the self-interaction singularity and it is not clear what the ``true'' microscopic theory approximating the Vlasov-Maxwell dynamics is supposed to be. The study of rigid charges (and their point-particle limit) thus seems like a natural way to make sense of the microscopic equations, with a longstanding tradition in the physical literature, see e.g. Lorentz 1892 \cite{Lorentz}, 1904 \cite{Lorentz2}, Sommerfeld, 1904 \cite{Sommerfeld}, Lyle, 2010 \cite{Lyle}. 
	
	Of course, the regularization thus introduced is still a technical expedient rather than a realistic physical theory. In particular, the $N$-dependence of the electron radius does not make much sense from a physical point of view and the Abraham model neglects spin as well as other possible effects due the extension of the charges. However, as other authors have pointed out before (see e.g. \cite{Golse, Elskens}), any more satisfying approach to the Vlasov-Maxwell dynamics will most likely require a satisfying solution to the self-interaction problem first. Given the current state of affairs, we believe that the result presented here constitutes significant progress with regard to the validity problem of Vlasov-Maxwell.

\newpage
\bibliography{VPlit}
\bibliographystyle{plain}
 \end{document}